\documentclass[a4paper,reqno]{amsart}

\pdfoutput=1

\usepackage{amssymb}
\usepackage[unicode=true,hidelinks]{hyperref}

\usepackage{amsthm}
\usepackage{caption}
\usepackage{mathtools}
\usepackage{tikz}
\usetikzlibrary{arrows}
\usetikzlibrary{positioning}
\usetikzlibrary{calc}
\usetikzlibrary{arrows.meta}

\usepackage{enumitem}
\usepackage{xparse}
\usepackage{wasysym}

\newcommand{\Scircle}{\mathbb S^1}
\newcommand{\Sloop}{\mathsf{loop}}
\newcommand{\inl}{\mathsf{inl}}
\newcommand{\inr}{\mathsf{inr}}

\usepackage{float}
\floatstyle{boxed} 
\restylefloat{figure}

\usepackage[normalem]{ulem}
\usepackage[capitalize]{cleveref}

\usepackage{upgreek}
\usepackage{scalerel}

\newcommand{\UU}{\mathcal U}
\newcommand{\CC}{\mathcal C}
\newcommand{\opp}[1]{{#1}^{\mathsf{op}}}
\newcommand{\CCop}{\opp\CC}
\newcommand{\DD}{\mathcal D}
\newcommand{\defeq}{\vcentcolon\equiv}
\newcommand{\Type}{\UU}

\newcommand{\blank}{\mathord{\hspace{1pt}\_\hspace{1pt}}} 

\newcommand{\Set}{\mathsf{Set}}
\newcommand{\ra}{\to}

\newcommand{\id}{\mathsf{id}}
\newcommand{\ass}{\mathsf{assoc}}
\newcommand{\idl}{\mathsf{idl}}
\newcommand{\idr}{\mathsf{idr}}

\newcommand{\Con}{\mathsf{Con}}
\newcommand{\Ty}{\mathsf{Ty}}
\newcommand{\Sub}{\mathsf{Sub}}
\newcommand{\Tm}{\mathsf{Tm}}

\newcommand{\refl}{\mathsf{refl}}

\newcommand{\ap}{\mathsf{ap}}

\newcommand{\ext}{\triangleright}

\newcommand{\p}{\mathsf{p}}
\newcommand{\q}{\mathsf{q}}

\newcommand{\substT}[1]{[#1]^\mathsf{T}}
\newcommand{\substt}[1]{[#1]^\mathsf{t}}
\newcommand{\emptycon}{\scaleobj{.75}\bullet}
\newcommand{\emptysubst}{\upepsilon}
\newcommand{\comma}{,}

\newcommand{\N}{\mathbb{N}}

\newcommand{\sN}{\mathbb{N}^{\mathsf{s}}}

\makeatletter
\DeclareRobustCommand{\ct}{\mathbin{\mathpalette\morphic@sqcdot\relax}}
\newcommand{\morphic@sqcdot}[2]{%
	\sbox\z@{$\m@th#1\centerdot$}%
	\ht\z@=.33333\ht\z@
	\vcenter{\box\z@}%
}
\makeatother

\newcommand{\isset}{\mathsf{isSet}}
\newcommand{\isprop}{\mathsf{isProp}}
\newcommand{\Prop}{\mathsf{Prop}}

\newcommand{\scomp}{\diamond}
\newcommand{\mc}{\scomp} 
\newcommand{\fc}{\circ} 

\newcommand{\isid}{\mathsf{isId}}

\newcommand{\ob}{\mathsf{Ob}}
\renewcommand{\hom}{\mathsf{Hom}}
\newcommand{\iseqv}{\mathsf{iseqv}}

\tikzset{
 boxedstuff/.style={
  rectangle,
  rounded corners,
  draw=black, thick,
  minimum height=2em,
  minimum width=5em,
  inner sep=5pt,
  },
 tikzincludebackground/.style={
  right hook->,
  shorten >=0.2cm,
  shorten <=0.2cm, 
  line width = 0.5pt,
  },
 tikzproject/.style={
  ->>,
  shorten >=0.2cm,
  shorten <=0.2cm, 
  line width = 0.5pt,
  },
 tikzpropproject/.style={
  >->>,
  shorten >=0.2cm,
  shorten <=0.2cm, 
  line width = 0.5pt,
  },
 tikzshortarrow/.style={
  shorten >=0.2cm,
  shorten <=0.2cm, 
  thick,
  },
 tikzequiv/.style={
  <->,
  shorten >=0.2cm,
  shorten <=0.2cm, 
  line width = 0.8pt,
  sloped, 
  above,
  },
 tikzbackgroundequiv/.style={
  <->,
  shorten >=0.2cm,
  shorten <=0.2cm, 
  line width = 0.5pt,
  right,
  },
}
\theoremstyle{plain}
\newtheorem{theorem}{Theorem}

\newtheorem{corollary}[theorem]{Corollary}
\newtheorem{lemma}[theorem]{Lemma}

\theoremstyle{definition}

\newtheorem{remark}[theorem]{Remark}

\newtheorem{example}[theorem]{Example}
\newtheorem{definition}[theorem]{Definition}

\newcommand{\steq}{=^\mathsf{s}}

\newcommand{\functor}{\xRightarrow{\textit{s}}}
\newcommand{\rffunctor}{\xRightarrow{\textit{RF}}}

\newcommand{\deltop}{\Delta_+^\mathsf{op}}
\newcommand{\finitedeltop}[1]{\left(\Delta_+^{\leq #1}\right)^\mathsf{op}}

\newcommand{\simplex}[1]{\Delta[{#1}]}
\newcommand{\horn}[2]{\Lambda^{#2}[{#1}]}

\newcommand{\iscontr}{\mathsf{isContr}}
\newcommand{\Acat}{\mathcal A}

\newcommand{\TT}{\mathcal T}
\newcommand{\TTp}{\mathcal T^\bullet}
\newcommand{\unit}{\mathsf{Unit}}
\newcommand{\bool}{\mathsf{2}}
\newcommand{\nerve}{\mathsf{N}_+}
\newcommand{\nervep}{\mathsf{N}_+^\bullet}
\newcommand{\Fin}{\mathsf{Fin}}

\usepackage[most]{tcolorbox}

\newtcolorbox{takeaway}{%
    enhanced jigsaw,
    colback=white,
    borderline={1pt}{-2pt}{black},
    boxrule=.8pt,
}

\newtcolorbox{problembox}{%
    sharp corners,
    colback=white,
    borderline={1pt}{-2pt}{black},
    boxrule=.8pt,
}

\author{Nicolai Kraus}
\title[Internal $\boldsymbol{\infty}$-Categorical Models of Dependent Type Theory]{Internal $\boldsymbol{\infty}$-Categorical Models \\ of Dependent Type Theory \\[.2cm] \small Towards 2LTT Eating HoTT}
\thanks{This work has been supported by the Royal Society, grant No.~URF\textbackslash R1\textbackslash 191055,
and the European Union, co-financed by the European Social Fund (EFOP-3.6.2-16-2017-00013,
Thematic Fundamental Research Collaborations, Grounding Innovation in Informatics and
Infocommunication)}

\begin{document}

\maketitle

\begin{abstract}
Using dependent type theory to formalise the syntax of dependent type theory is a very active topic of study
and goes under the name of ``type theory eating itself'' or ``type theory in type theory.''
Most approaches are at least loosely based on Dybjer's \emph{categories with families (CwF's)}
and come with a type $\Con$ of contexts, a type family $\Ty$ indexed over it modelling types, and so on.
This works well in versions of type theory where the principle of \emph{unique identity proofs} (UIP) holds.
In homotopy type theory (HoTT) however, it is a long-standing and frequently discussed open problem whether the type theory ``eats itself'' and can serve as its own interpreter.
The fundamental underlying difficulty seems to be that categories are not suitable to capture a type theory in the absence of UIP. 

In this paper, we develop a notion of \emph{$\infty$-categories with families ($\infty$-CwF's)}.
The approach to higher categories used relies on the previously suggested \emph{semi-Segal types},
with a new construction of identity substitutions that allow for both univalent and non-univalent variations.
The type-theoretic universe as well as the internalised (set-level) syntax are models, although it remains a conjecture that the latter is initial.
To circumvent the known unsolved problem of constructing semisimplicial types,
the definition is presented in \emph{two-level type theory} (2LTT).


Apart from introducing $\infty$-CwF's, the paper explains the shortcomings of 1-categories in type theory without UIP as well as the difficulties of and approaches to internal higher-dimensional categories.
\end{abstract}

\tableofcontents

\section{Introduction: Formalising Type Theory}

\newcommand\spaceheadlines{-3cm}
\newcommand\spacenext{.04cm}
\newcommand\spaceover{-2.55cm}

\begin{figure}
    \begin{alignat}{10}
      &&&&& \hspace*{\spaceheadlines} \textit{a semicategory of contexts and substitutions:} \nonumber \\
      & \Con && : &\quad & \Type \label{gat:Con'} \\
      & \Sub && : && \Con \ra\Con \ra \Type   \\
      & \blank\scomp\blank && : && \Sub\,\Theta\,\Delta\ra\Sub\,\Gamma\,\Theta\ra\Sub\,\Gamma\,\Delta   \\
      & \ass && : && (\sigma \scomp \delta) \scomp \nu = \sigma \scomp (\delta \scomp \nu)   \\[\spacenext]
      &&&&& \hspace*{\spaceheadlines} \textit{identity morphisms as identity substitutions:} \nonumber \\
      & \id && : && \Sub\,\Gamma\,\Gamma  \\
      & \idl && : && \id\scomp\sigma = \sigma   \\
      & \idr && : && \sigma\scomp\id = \sigma   \\[\spacenext]
      &&&&& \hspace*{\spaceheadlines} \textit{a terminal oject as empty context:} \nonumber \\
      & \emptycon && : && \Con \\
      & \emptysubst && : && \Sub\,\Gamma\,\emptycon   \\
      & {\emptycon\upeta} && : && \forall (\sigma : \Sub\,\Gamma\,\emptycon). \; \sigma = \emptysubst  \\[\spacenext]
      &&&&& \hspace*{\spaceheadlines} \textit{a presheaf of types:} \nonumber\\
      & \Ty  && : && \Con\ra\Type \\
      & \blank \substT\blank && : && \Ty\,\Delta\ra\Sub\,\Gamma\,\Delta\ra\Ty\,\Gamma \\
      & \substT\id && : && A\substT\id = A \\
      & \substT\scomp && : && A\substT{\sigma \scomp \delta} = A\substT\sigma\substT\delta  \\[\spacenext]
      &&&&& \hspace*{\spaceheadlines} \textit{a (covariant) presheaf on the category of elements as terms:} \nonumber\\
      & \Tm  && : && (\Gamma:\Con)\ra\Ty\,\Gamma\ra\Type  \label{gat:Tm'} \\
      & \blank \substt\blank && : && \Tm\,\Delta\,A\ra(\sigma:\Sub\,\Gamma\,\Delta)\ra \Tm\,\Gamma\,(A\substT\sigma)  \label{gat:subtm'}\\
      & \substt\id && : && t\substt\id = t &&  \hspace*{\spaceover} \text{over } \substT\id \\
      & \substt\scomp && : && t\substt{\sigma \scomp \delta} = t\substt\sigma\substt\delta &&  \hspace*{\spaceover} \text{over } \substT\mc \\[\spacenext]
      &&&&& \hspace*{\spaceheadlines} \textit{context extension, modelled by representability:} \nonumber \\
      & \blank\ext\blank && : && (\Gamma:\Con)\ra\Ty\,\Gamma\ra\Con  \\
      & \p && : && \Sub\,(\Gamma\ext A)\,\Gamma   \label{gat:pi1'} \\
      & \q && : && \Tm\,(\Gamma\ext A)\,(A\substT\p)   \\
      & \blank \, \comma\blank && : && (\sigma:\Sub\,\Gamma\,\Delta)\ra\Tm\,\Gamma\,(A\substT\sigma)\ra \Sub\,\Gamma\,(\Delta\ext A) \\
      & {\ext\upbeta_1} && : && \p\scomp(\sigma\comma t) = \sigma    \\
      & {\ext\upbeta_2} && : && \q \substt{\sigma\comma t} = t &&  \hspace*{\spaceover}  \text{over } \substT\mc \text{ and } {\ext\upbeta_1}  \\
      & {\ext\upeta} && : && (\p \comma \q) = \id && \\
      & {\comma\scomp} && : && (\sigma,t)\scomp\nu = (\sigma\scomp\nu\comma t\substt\nu) &&  \hspace*{\spaceover}  \text{over } \substT\mc 
    \end{alignat}
    \caption{The formulation of a category with families (CwF) as a generalised algebraic theory (GAT), as given by Kaposi and others~\cite{KapKovKra:embedding,kaposi_et_al:LIPIcs:2019:10532}.
    Presentation-wise, it slightly differs from (but is equivalent to) the similar suggestion by Altenkirch and Kaposi~\cite{alt-kap:tt-in-tt}.
    Above, the components are reordered and regrouped to make the connection to CwF's more visible.
    \emph{Over} refers to \emph{substitution} (a.k.a.\ \emph{transport}) in the terminology of the HoTT book~\cite[Chp.~2.3]{hott-book}. Implicit arguments are omitted. $\UU$ is a universe at any level.} \label{fig:Ambrus-changed-reordered}
\end{figure}

Dependent type theory in the style of
Martin-L\"of 
forms the foundation of various dependently typed programming languages and proof assistants, such as Agda, Coq, Epigram, Idris, and Lean.
Numerous variations of type theory have been considered,
and models are studied in order to better understand the properties of these theories.
Even the study and formalisation of models of type theory \emph{in type theory itself}
is an active field of research, involving various different models such as the setoid model implemented by Palmgren~\cite{palmgren2019type} and others,
parametric models of type theory (Bernady et al.~\cite{BERNARDY201567}),
or an implementation of the groupoid interpretation (Sozeau and Tabareau~\cite{sozeau2014internalization}) by Hofmann and Streicher~\cite{hofmannStreicher_groupoids}.

In particular formalisations of the syntax of type theory, i.e.\ the \emph{syntactic model} or \emph{term model}, have received significant attention.
Statements of the form ``type theory should be able to handle its own meta-theory'' are often made (e.g.\ the very first sentence by Abel et al.~\cite{abel2017decidability}), which refers to formalising the (intended) \emph{initial} model of type theory.%
\footnote{It is often seen as type-theoretic folklore that the initial model of a type theory coincides with the syntax,
i.e.\ that the syntactic model is initial.
While a proof for the calculus of constructions has been known for a while (cf.~Streicher's work~\cite{Streicher93}), precise and formalised proofs for intensional Martin-L\"of type theory became available only recently (cf.\ the talks presenting work by de Boer, Brunerie, Lumsdaine, and M\"ortberg~\cite{brunerie_talk,lumsdaine_talk,lums_bru_talk_in_Stockholm,lums_bru_talk_at_hottest} as well as the licanciate thesis by de Boer~\cite{deBoer:licentiate-thesis}.}

Altenkirch and Kaposi~\cite{alt-kap:tt-in-tt} call this simply \emph{type theory in type theory}.
To avoid confusion, we refer to the type theory in which these models are implemented as the \emph{host theory}, and the structure that get implemented as the \emph{object theory} or simply the \emph{model}.
The expression ``type theory eats itself'' for this concept was suggested by Chapman~\cite{chapman2009type},
inspired by Danielsson~\cite{danielsson2006formalisation} and others.
Other recent work on the same goal with various techniques
includes the studies by Escard\'o and Xu~\cite{autophagia}, the PhD thesis by Kaposi~\cite{ambrus:thesis},
the work by 
Gylterud, Lumsdaine, and Palmgren~\cite{lums_talk_dtt_in_dtt}, and the presentation by Buchholtz~\cite{buch_talk_dtt_in_dtt}.

The closely related idea to formalise the notion of a model of dependent type theory inside dependent type theory itself goes back at least to Dybjer's internal type theory via \emph{categories with families}, often just referred to as \emph{CwF's} \cite{dybjer1995internal} (see also formulation by Awodey~\cite{awodey_natural}).
Formalisations of various models have since then been pursued many times; a careful comparison of different definitions of models inside type theory has been given by Ahrens et al.~\cite{ahrens_et_al:LIPIcs:2017:7696,lmcs:4814}.
A category with families consists of a category 
with a terminal object, a presheaf of families on it, and a context extension operation. The 
objects of the category 
are
usually denoted by $\mathit{Con}$ (``contexts'') and the morphisms by $\mathit{Sub}$ (``substitutions'' or ``context morphisms'').
Many presentations further split the presheaf into two functors $\mathit{Ty}$ and $\mathit{Tm}$, with the idea being that $\mathit{Ty}$ gives the types in a context and $\mathit{Tm}$ the terms of a type.
The terminal object represents the empty context.

CwF's can easily be presented as a \emph{generalised algebraic theory} (GAT) as introduced by Cartmell~\cite{cartmell:thesis,cartmell1986generalised}.
A GAT consists of \emph{sorts}, \emph{operations}, and \emph{equations}.
In the case of CwF's, the sorts are contexts, substitutions, types, and terms.
Examples for operations are composition of substitutions,
identity substitutions, and context extension.
The equalities include the associativity law for composition, identity laws, laws expressing that types and terms form functors, and so on.
In type theory, sorts can be represented as types and type families.
This includes a type $\Con : \UU$ for contexts and families $\Sub : \Con \to \Con \to \UU$ as well as $\Ty : \Con \to \UU$ for substitutions and types.%
\footnote{$\UU$ is a universe. We omit universe indices everywhere and implicitly use universe polymorphism.}
Operations are functions with sorts as codomains.
For example, given a type $A : \Ty \, \Delta$ and a substitution $\sigma : \Sub \, \Gamma \, \Delta$, we need an operation (often written as $\blank [ \blank ]$) which gives us a new type $A[\sigma] : \Ty \, \Gamma$.
Equalities are stated using Martin-L\"of's \emph{identity type}, also known as \emph{equality type}, \emph{path type}, or \emph{identification type}, denoted by $a = b$.%
\footnote{At this point, it is important to emphasis the difference between the internal equality type $a = b$ and the meta-theoretic \emph{judgmental equality}, also known as \emph{definitional equality}, written as $a \equiv b$.}
As an example, 
we need an equality of the form $(A[\sigma])[\tau] = A[\sigma \mc \tau]$.
A full presentation of type theory in this form, similar to the one developed by Altenkirch and Kaposi~\cite{alt-kap:tt-in-tt}, is shown in
\cref{fig:Ambrus-changed-reordered}.%
\footnote{Regarding notation: 
    In the host theory, we use Agda-style notation for dependent function types and write $(x : A) \to B \, x$ instead of $\Pi(x:A). B(x)$. Dependent pair types are denoted as $\Sigma(x:A). B(x)$.
    We reserve the symbol $\fc$ for function composition in the host theory and denote composition of substitutions (context morphisms) in the model by $\mc$.
    Implicit arguments in the host theory are denoted by $\{x : A\}$, but they are completely omitted in \cref{fig:Ambrus-changed-reordered} for readability.}
If one works in a type theory that has \emph{quotient inductive-inductive} \cite{Altenkirch2018} or \emph{higher inductive-inductive types} \cite{kaposi2019signatures}, then
defining a CwF as a GAT automatically gives a formalisation of the initial model (the intended syntax).

Besides the initial model, a very important interpretation is the \emph{standard model} (see e.g.\ the work by Altenkirch and Kaposi~\cite{alt-kap:tt-in-tt} and Coquand et al.~\cite{coquand_et_al:LIPIcs:2019:10518}), sometimes known as the \emph{meta-circular interpretation}.
It is based on the observation that the category of [small] types can be equipped with the structure of a CwF in a canonical way.
In detail, contexts and substitutions in the standard model are [small] types and functions of the host theory, 
types of the model are dependent types of the host theory, and terms are dependent functions.

We have a morphism of models from the initial model to any other model, and in particular to the standard model.
This means that a universe in the host theory is a model of type theory itself.
From a programming perspective, this is interesting as it automatically gives us an interpreter:
We formalise the syntax of type theory inside type theory, and anything we do with this syntax can be interpreted as a construction in the host theory.
As Shulman~\cite{shulman:eating} argues, we should expect that this is possible if we want to view type theory as a general purpose programming language, as any such language should be able to implement its own interpreter or compiler.
This motivates the formulation of the following somewhat vague task: 

\begin{problembox}
\noindent
\textbf{General Problem:} In a (given specific version of) dependent type theory, define the notion of \emph{model} such that both the initial model and the standard model can be constructed, and such that the initial model can reasonably be viewed as syntax.
\end{problembox}

\noindent
It seems difficult, but also slightly besides the point, to make the above \emph{General Problem} completely precise to exclude trivial solutions; a satisfactory solution is something of the category ``we know it when we see it''.
The expectation is that a \emph{model} should 
allow at least the constructions listed in \cref{fig:Ambrus-changed-reordered}, i.e.\ each component of that figure should be definable.
An (algebraic) definition of a model will give rise to the definitions of \emph{model morphism} and thus \emph{initiality} in a canonical way.
Since we don't want the initial model to be trivial, the notion of model should contain base types; these are omitted in \cref{fig:Ambrus-changed-reordered} (and in most parts of the current paper) for simplicity, but easy to add, cf.~\cite{alt-kap:tt-in-tt}.

The most interesting question is arguable what it means that the initial model can be viewed as syntax.
Note that the equality type of the host theory plays the role of definitional/judgmental equality in the (initial) model.
Therefore, we expect that the equality of the initial model is \emph{proof-irrelevant}, i.e.\ $\Con$, $\Sub$, $\Ty$, $\Tm$ are \emph{(homotopy) sets} in the terminology of homotopy type theory (but see footnote \cref{finster-footnote} on page \pageref{finster-footnote}).

If we want to model an intensional type theory (which is the primary case of interest here), it is further desirable that the equality in the initial model is \emph{decidable}, corresponding to decidable type checking.
This means for example that, for $\Gamma : \Con$, we have
\begin{equation}
(A \, B  : \Ty \, \Gamma) \to (A = B) + \neg (A = B). 
\end{equation}
If we instead model extensional Martin-L\"of type theory, we cannot expect this (cf.\ Clairambault and Dybjer's work~\cite{clairambault_dybjer_2014}), but the requirement of proof-irrelevance remains.

In intensional dependent type theory with the axiom of \emph{unique identity proofs (UIP)}, i.e.\ the assumption that \emph{all} equalities are proof-irrelevant,
and quotient inductive-inductive types~\cite{Altenkirch2018}, a solution to the problem has been given by Altenkirch and Kaposi. They construct the initial model with the standard interpretation \cite{alt-kap:tt-in-tt} and show that the initial model has decidable equality \cite{kaposi2017normalisation}.

We are interested in settings where UIP is not assumed, which makes the problem much harder.
In particular, it is a long-standing open problem whether homotopy type theory (HoTT) \cite{hott-book}, which rejects UIP, is able to handle its own meta-theory.
The question was first raised by
Shulman~\cite{shulman:eating} and has since then been discussed very actively, in particular by
Escard\'o and Xu~\cite{autophagia},
Gylterud, Lumsdaine, and Palmgren~\cite{lums_talk_dtt_in_dtt},
Buchholtz~\cite{buch_talk_dtt_in_dtt},
and Altenkirch~\cite{alt_nicolaitalk}.
Progress in the setting without UIP has also been made by Abel, \"Ohman, and Vezzosi~\cite{abel2017decidability}.
The core task, as described by the \emph{General Problem} above, is still unsolved.

If we define a CwF in homotopy type theory to consist of the components in \cref{fig:Ambrus-changed-reordered}, maybe with additional base types, universes, or $\Pi$- and $\Sigma$-types as suggested for example by Kaposi, Huber, and Sattler~\cite{kaposi_et_al:LIPIcs:2019:10532}, we have to decide how we react to the absence of UIP.
If we simply ignore it, the equality types will not be well-behaved and expected properties will not hold, a very common phenomenon in homotopy type theory.
In particular, the initial model will not have proof-irrelevant (or even decidable) equality, making it unsuitable as ``syntax'' in any form. Phrased differently, the expected quotiented term model will not be initial.
The typical strategy to address this problem is to explicitly include a condition which ensures that all involved types are truncated at a certain level (i.e.\ that equality has to be irrelevant at some level).
As a prominent example, the notion of \emph{category} considered by Ahrens, Kapulkin, and Shulman~\cite{ahrens_rezk} requires morphisms to form a set to ensure well-behavedness.
Unfortunately, this would not be a satisfactory solution in our situation.
Since univalent universes in homotopy type theory are provably not set-truncated \cite[Ex.~3.1.9]{hott-book}, any such truncatedness assumption would mean that the standard ``model'' will cease to be a model.

The fundamental underlying issue is that (ordinary) categories are not suitable to talk about non-truncated structures internally in homotopy type theory or other versions of dependent type theory without UIP.
The \emph{laws} of category theory, such as associativity or identity laws, are not appropriately modelled by the internal equality type.
In fact, equalities in a theory without UIP are \emph{data} rather than \emph{properties}.
Thus, an internal equality expressing associativity is much more similar to an isomorphism in a bicategory than a law in a (1-)category.
Being even more precise, we should say that an equality behaves like an isomorphism in an $(\infty,1)$-category (for simplicity called $\infty$-category; a survey has been given by Bergner~\cite{bergner2010survey}), as higher equalities
of any levels will be possibly non-trivial data (cf.\ the work by Sattler and the current author~\cite{krausSattler_universes}).

This motivates the current paper.
We define the notion of a model of type theory using $\infty$-categories rather than ordinary categories.
In a nutshell, doing so corresponds to adding all higher coherences which are missing from \cref{fig:Ambrus-changed-reordered}, turning it into an \emph{$\infty$-category with families}.
The main difficulty is that $\infty$-CwF's need an infinite tower of data, and this data needs to be organised appropriately. 
It is in particular unknown 
whether this structure can be represented in the version of homotopy type theory developed in the standard reference, the ``HoTT book''~\cite{hott-book}.
To circumvent this problem, we work in a \emph{two-level type theory} \cite{voe_hts,ann-cap-kra:two-level} which makes the construction possible.
This nevertheless has significant implications for ``ordinary'' homotopy type theory, where
all finite special cases \emph{can} be expressed. 
Therefore, our construction automatically gives a definition in HoTT of $(2,1)$-CwF's, $(3,1)$-CwF's, and more generally $(n,1)$-CwF's for any externally fixed natural number $n$.
This already strongly generalises ordinary CwF's, which are merely $(1,1)$-CwF's in this scheme.

The specific approach to $\infty$-categories that we work with 
is based on the \emph{semi-Segal types} suggested independently by Capriotti~\cite{paolo:thesis} and Schreiber~\cite{Schreiber2012}, further studied by the current author together with Capriotti~\cite{capKra_semisegal} as well as with Annenkov, Capriotti, and Sattler~\cite{ann-cap-kra:two-level}.
While this approach immediately gives a definition of an \emph{$\infty$-semicategory}, adding the identities in a well-behaved manner is trickier than one might expect.
All the authors of the work mentioned above use identities which have \emph{univalence} built-in.
This is natural in homotopy type theory, but not necessarily desirable when considering models since we would not expect the syntactic model to be univalent (and even if we do, previous formalisations of the syntax are not univalent and would not be captured by a formulation that requires univalence).
We therefore introduce a new definition of identities, namely \emph{idempotent equivalences}.
We prove that our identity structure is unique (i.e.\ a propositional property rather than data) and therefore automatically fully coherent.
A further propositional property expressing univalence can be added optionally on top.
In order to define $\infty$-CwF's, we naturally have to develop some notions of $\infty$-category theory inside type theory such as $\infty$-functors or the $\infty$-category of elements of an $\infty$-presheaf.

We prove that the natural constructions and examples which are broken for ordinary CwF's (based on 1-categories) work when we move to $\infty$-categories, indicating that the latter are much more well-behaved.
We also discuss three additional properties which an $\infty$-CwF can have. First, an $\infty$-CwF can be purely based on sets (i.e.\ on types with proof-irrelevant equality), such as the syntax.
Second, an $\infty$-CwF can be univalent, such as the standard model.
And third, it can be \emph{finite-dimensional}, such as the higher CwF of $n$-types.
The latter case can be formulated in homotopy type theory, without the need for 2LTT.

\subsubsection*{Related work}

The original connection between higher categories (in the form of spaces) and type theory, discovered by Awodey and Warren~\cite{awodeyWarren_HTmodelsOfIT} and by Voevodsky (see the presentation by Kapulkin and Lumsdain~\cite{kap-lum:simplicial-model}), marks the beginning of the study of homotopy type theory and univalent foundations.
While Voevodsky's model of the univalence axiom uses simplicial structures, the superficial similarity to the work in this paper may be somewhat misleading;
in fact, the motivation for the appearance of higher categories is reversed.
Voevodsky's simplicial set model uses higher categories in order to model the equality type of the type theory that \emph{is} modelled, while we are using higher categories because the \emph{host theory} is unsuitable for working with ordinary categories.
While Kapulkin and Lumsdaine~\cite{kap-lum:simplicial-model} define one particular model, the current paper defines a type (``classifier'') of models.
Similarly, Barras, Coquand, and Huber~\cite{barras_coquand_huber_2015} construct a (concrete) model interpreting types as semisimplicial sets.
Boulier~\cite{simonboulier:thesis} and other authors discuss and formalise various models assuming UIP locally or globally.
Abel, \"Ohman, and Vezzosi~\cite{abel2017decidability} formalise the syntax in a type theory
without UIP, but their induced notion of model includes non-well-typed components which are not present in a type-theoretic universe; thus, the standard interpretation is not a model, and their approach does not give a solution to the \emph{General Problem} discussed above.
Complete semi-Segal types are a known approach to univalent $\infty$-categories~\cite{paolo:thesis,ann-cap-kra:two-level,capKra_semisegal},
and the current paper uses a variation of those in order to avoid completeness/univalence.
Nguyen and Uemura~\cite{taichi_infty} propose the development of \emph{$\infty$-type theories}, where definitional equalities are replaced by homotopies; although their proposal is not worked out in full at the time of writing, we speculate that our $\infty$-CwF's are in principle general enough to also be a suitable notion of model for $\infty$-type theories.
Capriotti and Sattler~\cite{paoloChristian_segal} build an interpretation based on Segal types for a specific type theory and prove the equivalence between the initiality and induction principle for higher inductive-inductive types.

\subsubsection*{Overview of the paper}

\cref{sec:tt-in-tt} discusses the formalisation of 1-categorical models, in the form of CwF's, in a type theory with UIP.
This approach does not work anymore when UIP is dropped as we see in
\cref{sec:challenges}, where several examples that break without UIP are studied.
For our higher-categorical approach, we want to use semisimplicial types in two-level type theory, both introduced in \cref{sec:infstrucs}.
The heart of the paper is \cref{sec:highercwf}, where we develop enough internal $\infty$-category theory in order to define $\infty$-CwF's.
In \cref{sec:examples-of-iCwF}, we prove that all the examples that were broken without UIP work when using $\infty$-CwF's, and in \cref{sec:variations}, we discuss variations including univalent and finite dimensional CwF's.
In \cref{sec:openconclusions}, we discuss open problems and conjectures, and conclude.

\subsubsection*{Specification of the host type theory}

We work in dependent type theory (the ``host theory'') with $\Sigma$- and $\Pi$-types satisfying their respective judgmental $\eta$-rule, and with a hierarchy of universes.
We always assume function extensionality.
In \cref{sec:tt-in-tt}, we assume that the host theory satisfies UIP in order to discuss the standard approaches and why they fail without UIP.
From \cref{sec:challenges} on, we drop this assumption and work in homotopy type theory with all definitions and features introduced by the HoTT book~\cite{hott-book}.
In particular, we use the notions of \emph{proposition ($\mathit{-1}$-type)}, \emph{set ($\mathit{0}$-type)} and more generally \emph{$\mathit{n}$-type},
as well as the \emph{path over} notation expressing substitution/transport ($a =_p b$ means $\mathsf{subst}(p,a) = b$).

From \cref{sec:two-level} on, we assume that the host theory is equipped with a second kind of equality type, turning it into a two-level type theory.
This theory is further specified in the mentioned section.

\subsubsection*{Agda formalisation}
The discussed \emph{identities via idempotent equivalences}, which turn an $\infty$-semicategory into an $\infty$-category, are a small part of the paper but represent a core idea for the construction of $\infty$-CwF's.
It turns out that the idea and all proofs can be formulated in a very simple setting that can be formalised without having to refer to 2LTT.
This paper is supplemented by an Agda formalisation of the development.%
\footnote{Source code: \url{https://github.com/nicolaikraus/idempotentEquivalences}; clickable \texttt{html}: \url{https://nicolaikraus.github.io/docs/html-idempotentequivalences/Identities.html}
\label{footnote:agda-links}}

\section{CwF's: 1-Categorical Models of Dependent Type Theory} \label{sec:tt-in-tt}

\subsection{Motivation and categorical definition}

In order to define higher- or $\infty$-categorical internal models of type theory, we need to find a suitable one-dimensional formulation which can serve as a starting point for generalisations.
In the current section, we therefore want to work with types that do not possess non-trivial higher structure.
This means that all types in this section are assumed to satisfy the principle of \emph{unique identity proofs (UIP)}, i.e.\ are \emph{(homotopy) sets}.

In the introduction we have seen \cref{fig:Ambrus-changed-reordered} which presents the components of type theory 
as a generalised algebraic theory and is due to Kaposi and others~\cite{KapKovKra:embedding,kaposi_et_al:LIPIcs:2019:10532}.
In detail, a \emph{model of type theory} according to \cref{fig:Ambrus-changed-reordered} consists of four types and type families ($\Con$, $\Ty$, $\Sub$, $\Tm$), together with ten terms that inhabit the four type families in various ways, and together with twelve equations.
The original work by Altenkirch and Kaposi~\cite{alt-kap:tt-in-tt} (as well as \cite{kaposi_et_al:LIPIcs:2019:10532,KapKovKra:embedding}) contains additional components for various type formers. These and the twenty-six components of \cref{fig:Ambrus-changed-reordered}  are then viewed as signatures and constructors of a \emph{quotient inductive-inductive type (QIIT)} which defines the initial model (the ``syntax'').
In the current paper, we are interested in models in general, not only the initial such model. Moreover, the initial model of \cref{fig:Ambrus-changed-reordered} is rather uninteresting since we have not added base types, meaning that the initial model has only the empty context and no types.

A more compact description of a model of type theory is given by the categorical formulation of a CwF \cite{dybjer1995internal}.
Before stating its definition, let us recall the following basic categorical constructions.
Let $\mathcal D$ be a category and $F$ be a functor from $\mathcal D$ to the category $\Set$ of sets and functions.
The \emph{category of elements} is denoted by $\int_{\mathcal D} F$. 
Its objects are the pairs $\{(d,x) \mid d \in \mathcal D_0, x \in F(d)\}$ and a morphism from $(d,x)$ to $(e,y)$ is a morphism $f \in \mathcal D_1(d,e)$ such that $F(f)(x) = y$.
For $d \in \mathcal D_0$, the \emph{slice category} $\mathcal D \slash d$ has as objects those morphisms of $\mathcal D$ which have $d$ as codomain, while morphisms are those of $\mathcal D$ which make the evident triangle commute.
Finally, the functor $F$ is \emph{represented} by $d \in \mathcal D_0$ if $F$ is naturally isomorphic to $\mathcal D_1(d, \blank )$.
By the Yoneda lemma, this is equivalent to having an element $x \in F(d)$ such that $(d,x)$ is initial in $\int_{\mathcal D} F$.
We follow standard terminology and call $d$ the \emph{representing object} and $x$ \emph{universal element}.

\begin{definition}[CwF] \label{def:CwF}
 A \emph{category with families (CwF)} is given by:
 \begin{enumerate}[label=(\roman*)]
  \item a category $\CC$ with a terminal object,
  \item a presheaf $\Ty : \CCop \to \Set$; we will write $A \substT \sigma$ instead of $\Ty(\sigma)(A)$; 
  \item a functor $\Tm : \left(\int_{\CCop} \Ty\right) \to \Set$; for $t \in \Tm(\Gamma,A)$, we will write $t \substt{\sigma}$ instead of $\Tm(\sigma)(t)$,
  \item and, for every $\Delta \in \CC_0$ and $A \in \Ty(\Delta)$, an object $\Delta \ext A$ together with a morphism $p_A : \CC_1(\Delta \ext A, \Delta)$ which represents the functor 
 \begin{align}
  & (\CC \slash \Delta)^\mathsf{op} \to \Set \\
  & (\Gamma, \sigma) \mapsto \Tm (\Gamma, A \substT \sigma). \label{eq:needs-to-be-repr}
 \end{align}
 \end{enumerate}
\end{definition}

\subsection{CwF's as generalised algebraic theories} \label{subsec:GAT}

\cref{fig:Ambrus-changed-reordered} can be seen as a direct implementation of \cref{def:CwF} in type theory.
The correspondence is made clear by the sub-headlines in the figure.
Let us go through the points in \cref{def:CwF}:
 \begin{enumerate}[label=(\roman*)]
  \item The category $\CC$ in \cref{def:CwF} is the category consisting of the first two groups of data in \cref{fig:Ambrus-changed-reordered}; the reason for splitting this category in a semicategory and separate identities will become clear later.
  The terminal object is self-explanatory.
  \item The presheaf $\Ty$ is implemented in the straighforward way.
  \item For the functor $\Tm : \left(\int_{\CCop} \Ty\right) \to \Set$, 
the object part is directly given by \cref{gat:Tm'}.
The type-theoretic version of the morphism part makes use of a standard simplification:
Having two objects $(\Gamma, B)$ and $(\Delta, A)$ in $\int_{\CCop} \Ty$ together with a morphism $(\sigma : \CC_1(\Gamma, \Delta), e: B = \Ty(\sigma)(A))$ is as good as having only $\Gamma$, $\Delta$, $\sigma$, and $A$.
Thus, the morphism part of $\Tm$ can be stated as
\begin{equation}
 \begin{alignedat}{2}
  & \blank \substt{\blank } : & \quad &\{\Gamma \, \Delta : \Con\} \to \{A : \Ty(\Delta)\} \to \\
  &&& (\sigma : \Sub \, \Gamma \, \Delta) \to \Tm \, \Delta \, A \to \Tm \, \Gamma \, (A \substT \sigma)
 \end{alignedat}
\end{equation}
The version \eqref{gat:subtm'} omits the implicit arguments and swaps the explicit arguments.
  \item Given $\Delta$ and $A$, the pair $(\Delta \ext A, p)$ is the representing object, while $q$ is the universal element.
  Initiality of the object $(\Delta \ext A, p, q)$ in the category of elements of $\Tm$ translates in type theory to:
  For all $\sigma : \Sub \, \Gamma \, \Delta$ and $t : \Tm \, \Gamma \, (A \substT{\sigma})$, the type
  \begin{equation} \label{eq:contr-represent}
   \Sigma (\gamma : \Sub \, \Gamma \, (\Delta \ext A)). \Sigma (e : \sigma = p \mc \gamma). (q\substt{\sigma} =_e t)
  \end{equation}
  is contractible (i.e.\ has a unique inhabitant). 
  The inhabitant (``centre of contraction'') is given by $\gamma :\equiv (\sigma,t)$ together with the equations $\ext\upbeta_1$ and $\ext\upbeta_2$.
  Having shown that there is a suitable $\gamma$, we need to check that it is the only one; thus, assume we have $\gamma'$ satisfying the equations.
  It then follows that
  \begin{equation}
  \begin{alignedat}{12}
   &&&&& \gamma' \\
   &\textit{by $\mathsf{(idl)}$ and $(\ext\upeta)$} &\qquad & = & \quad & \gamma' \mc (p, q) \\
   &\textit{by $(\comma\scomp)$} && = && (p \mc \gamma' , q \substt {\gamma'}) \\
   &\textit{by assumption} && = && (\sigma , t).
  \end{alignedat}
  \end{equation}
  Thanks to the assumption that all involved types satisfy UIP, the equations are automatically unique.
  Vice versa, from contractibility of \eqref{eq:contr-represent}, one can derive the components $(\comma\scomp)$ and $(\ext\upeta)$.
 \end{enumerate}
 
\subsection{Examples} \label{sec:examples-of-cwfs}
 
The literature covers many examples for CwF's.
For us, the following standard and well-known examples are particularly interesting:

\begin{example}[syntax/initial model] \label{ex:initial}
 In a version of type theory with \emph{quotient inductive-inductive types (QIITs)} \cite{Altenkirch2018,kaposi_et_al:LIPIcs:2018:9190},
 we can read \cref{fig:Ambrus-changed-reordered} directly as a QIIT signature.
 There is an implied notion of model morphism, and the QIIT is automatically the initial such model.
 This is how the construction is presented by Altenkirch and Kaposi~\cite{alt-kap:tt-in-tt},
 and we refer to their construction as the \emph{syntax QIIT}.
 The initial CwF without base types consists of only the empty context, but they demonstrate that it is easy to add further components:
 $\mathsf{Unit}$, $\mathsf{Bool}$, $\Sigma$- and $\Pi$-types as well as a ``universe''
 $(\mathsf{U}, \mathsf{El})$.

 Kaposi and Altenkirch~\cite{kaposi2017normalisation} further show that the syntax QIIT has decidable equality.
 This is important if we want to view the initial model as formalised syntax, and from this point of view, the internal equality type plays the role of the meta-theoretic definitional equality, which in most type theories is required to be decidable.
\end{example}

\begin{example}[standard model] \label{ex:stdmodel}
 The \emph{standard model} is the CwF of types and functions: each expression in \cref{fig:Ambrus-changed-reordered} is interpreted by its canonical ``semantic counterpart''.
 We need a fixed [small] universe $\UU$ to define this internally.
 Then, the standard model is given by:

\noindent
\begin{minipage}[t]{.5\textwidth -.25cm}
 \begin{alignat}{3}
  & \Con \defeq \UU \\
  & \Sub \, \Gamma \, \Delta \defeq \Gamma \to \Delta  \\
  & \delta \mc \sigma \defeq \delta \fc \sigma \\
  & \id \defeq \lambda x.x\\
  & \emptycon \defeq \unit \\
  & \Ty \, \Gamma \defeq \Gamma \to \UU\\
  & A \substT \sigma \defeq A \fc \sigma 
 \end{alignat}
 \hfill
\end{minipage}
\begin{minipage}[t]{.5\textwidth -.25cm}
    \begin{alignat}{3}
    & \Tm \, \Gamma \, A  \defeq \Pi(x : \Gamma).(A \, x) \\
    & t \substt \sigma \defeq t \fc \sigma \\
    & \Gamma \ext A \defeq \Sigma (x : \Gamma).(A \, x) \\
    & \p \defeq \mathsf{proj}_1 \\
    & \q \defeq \mathsf{proj}_2 \\
    & (\sigma, t) \defeq \lambda x. (\sigma \, x , t \, x)\\
    & \ldots \nonumber
    \end{alignat}
\end{minipage}

\noindent
 Several authors (see e.g.\ \cite{alt-kap:tt-in-tt,autophagia}) have noticed that all equations in this model hold definitionally (i.e.\ by $\refl$): function composition is definitionally associative (assuming $\eta$ for $\Pi$-types), and so on.
 This will play a role later in the current paper.
\end{example}

\begin{example}[modelling an axiom] \label{ex:axmodel}
 Assume we are given a set-CwF $\CC$.
 This model may have additional types and type formers such as $\Pi$-types and universes.
 We now may want a model where a global assumption is satisfied, i.e.\ the axiom of function extensionality.
 Such a model can be created by fixing the context $\Gamma_0$ to contain a term which witnesses function extensionality, and constructing the slice $\CC \slash \Gamma_0$.
 As usual, contexts of $\CC \slash \Gamma_0$ are pairs $(\Delta : \Con, \delta : \Sub \, \Delta \, \Gamma_0)$, and a morphism between $(\Delta, \delta)$ and $(\Phi, \varphi)$ is a pair $(f : \Sub \, \Delta \, \Phi, e : \delta = \varphi \mc f)$.
 The other components of the new model are those given by $\CC$.
\end{example}

\begin{takeaway}
\noindent
\textbf{Section summary.}
Dybjer's notion of a model of type theory, a \emph{category with families (CwF)}, can be represented as a generalised algebraic theory and implemented in dependent type theory.
The most important examples of models are the \emph{initial model}, which one expects to coincide with an internal representation of the syntax, and the \emph{standard model} of types and functions, which provides the ``obvious semantics''. Both models have been studied and work well in a type theory with UIP.
\end{takeaway}

\section{Challenges in Type Theory without UIP} \label{sec:challenges}
 
As discussed in the previous section, there are clear and extensively studied strategies for the development of models of type theory via CwF's in a dependent type theory which satisfies the principle of unique identity proofs (UIP).

The purpose of the section is to demonstrate that $1$-categorical methods may be unsatisfactory in a setting without UIP.
The underlying question is: \emph{What is a good notion of a category?} 
We can add a truncation condition explicitly and accept that this may exclude cases that of interest, as it is done e.g.\ by Ahrens, Kapulkin, and Shulman~\cite{ahrens_rezk}.
If we choose not to add such a condition, the types may behave differently than expected.

In this section, we assume that we work in the version of homotopy type theory that is developed in the HoTT book~\cite{hott-book}.
In principle, the examples that we give apply in any version of dependent type theory which does not (or not necessarily) validate the principle UIP.

\subsection{Deficiency of Truncated Structure}

In order to reproduce the constructions from \cref{sec:tt-in-tt} in HoTT or other dependent type theories without UIP, we could simply modify the definition of a model (i.e.\ of a CwF) by adding the requirement that the involved types and type families 
are sets anyway:

\begin{definition}[set-CwF] \label{def:setCwF}
 A \emph{set-CwF} is a 30-tuple $(\Con, \Sub, \ldots)$ of the 26 components given in \cref{fig:Ambrus-changed-reordered}, together with four components:
 \begin{alignat}{2}
  & \mathsf{conset} : &\;\; & \isset(\Con) \\
  & \mathsf{tyset} : && (\Gamma : \Con) \to \isset(\Ty \, \Gamma) \\
  & \mathsf{subset} : && (\Gamma \, \Delta : \Con) \to \isset(\Sub \, \Gamma \, \Delta) \\
  & \mathsf{tmset} : && (\Gamma : \Con) \to (A : \Ty \, \Gamma) \to \isset(\Tm \, \Gamma \, A)
 \end{alignat}
\end{definition}

 \noindent
Formally, a set-CwF is an element of an iterated $\Sigma$-type, or, if supported by the theory, an element of a record type with 30 entries.

As long as we only care about the ``sub-theory of sets'', we can do everything with set-CwF's that we could do in the setting with UIP.
In particular, assuming the type theory supports QIITs, the initial set-CwF can still be defined via Altenkirch and Kaposi's \emph{syntax-QIIT} (c.f.\ \cref{ex:initial}).
The main drawback is the observation that the \emph{standard model} will no longer work:

\begin{example}[standard model; continuing \cref{ex:stdmodel}] \label{ex:stdmodel2}
 In HoTT, the [lowest or higher] universe $\UU$ is not a set. 
 Thus, the standard model presented in \cref{ex:stdmodel} is not a set-CwF.
\end{example}

A very weak version of the standard model is still possible:
We can define a ``universe of propositions'' as $\Prop \defeq \Sigma(X : \UU). \isprop(X)$. Since this is a set, it can be used to build a model as in \cref{ex:stdmodel}.
However, as explained in the introduction, the non-existence of the general ``standard model'' is unsatisfying.
One possible view is that, 
unlike many other programming languages, HoTT cannot serve as its own interpreter with this approach.

The situation does not become better if we work with \emph{univalent} $1$-categories \cite{ahrens_rezk} instead, where substitutions have to form a set and contexts a $1$-type.

\subsection{Deficiency of Wild/Incoherent Structure}

Another approach would be to ``ignore'' the fact that we do not have UIP, and simply use the definition which works well in a setting with UIP.

\begin{definition}[wild CwF] \label{def:wildCwF}
 A \emph{wild CwF} is a 26-tuple $(\Con, \Sub, \ldots)$ of all the components in \cref{fig:Ambrus-changed-reordered}.
 In other words, a wild CwF is a set-CwF without the four components in \cref{def:setCwF}.
\end{definition}

The universe forms a wild CwF:
the construction in \cref{ex:stdmodel} works word for word, since the assumption of UIP was never used.
Unfortunately, the absence of UIP breaks the other constructions that we have discussed in \cref{sec:examples-of-cwfs}.
Before we can see this, we record the following observation.
Making some of the ``implicit arguments'' precise, both $\idl_{\id(\Gamma)}$ and
$\idr_{\id(\Gamma)}$ are of type $\id(\Gamma) \mc \id(\Gamma) = \id(\Gamma)$.
The definition of a bicategory would include an equality ensuring these coincide, but this is not the case for wild CwF's:

\begin{lemma} \label{lem:nonset-wild-CwF}
    In homotopy type theory, we can construct a wild CwF $(\Con, \Sub, \ldots)$ with a context $\Gamma : \Con$ such that $\idl_{\id(\Gamma)} \neq \idr_{\id(\Gamma)}$.
\end{lemma}

\begin{proof}
    The wild CwF is given by slightly modifying the standard model (\cref{ex:stdmodel}).
    Let $\Con$ be the coproduct $\UU + \unit$, with the intuition that the single new context $(\inr \, \star)$ is a copy of the circle type $\Scircle$, the type with a non-trivial equality $\Sloop$.
    Thus, we extend all components from \cref{ex:stdmodel} accordingly,
    \begin{alignat}{3}
    & \Sub \, (\inl \, X) \, (\inr \, \star) &\; &\defeq  &\quad & X \to \Scircle \\
    & \Sub \, (\inr \, \star) \, (\inl \, Y) &&\defeq  && \Scircle \to Y \\
    & \Sub \, (\inr \, \star) \, (\inr \, \star) &&\defeq  && \Scircle \to \Scircle \\
    & \Ty \, (\inr \, \star) && \defeq && \Scircle \to \UU
    \end{alignat}    
    and so on.
    We choose all equalities to be constantly $\refl$ as in \cref{ex:stdmodel} apart from $\idl$, which we set as follows:
    \begin{itemize}
        \item for $\sigma : \Sub \, \Gamma \, (\inl \, X)$, we set $\idl_{\sigma} \defeq \refl$;
        \item for $\sigma : \Sub \, \Gamma \, (\inr \, \star)$, we set
        $\idl_{\sigma} \defeq \mathsf{funext}(\lambda x. \ell(\sigma(x)))$, where $\ell$ is the function that gives us a non-trivial auto-equality for every element of the circle \cite[Lem 6.4.2]{hott-book}.
    \end{itemize}
    Then, we have $\idl_{\id(\inr \, \star)} \neq \idr_{\id(\inr \, \star)}$ as required.  
\end{proof}

This allows us to see that the initial wild CwF does not behave as expected:

\begin{example}[initial model; continuing \cref{ex:initial}] \label{ex:initial2}
 In order for the initial model to be non-trivial, let us assume that we add a base type or a family of base types to the specification of a wild CwF as discussed in \cref{ex:initial}.
 Altenkirch and Kaposi's \emph{syntax QIIT} is a wild CwF and remains a valid example.
 However, it is not the \emph{initial} wild CwF:
 We can take the distinguished version of $\Scircle$ in the construction of \cref{lem:nonset-wild-CwF} as a base type (or family of base types),
 and if the \emph{syntax QIIT} was initial, it would come with a map to this modified standard model.
 As the distinguished context would be in the image of this map,
 the syntax QIIT would have the property
 $\idl_{\id} \neq \idr_{\id}$, contradicting the fact that the syntax QIIT is a set-CwF. 
 
 Assuming we work in homotopy type theory with higher inductive-inductive types (the not-necessarily-truncated version of QIITs), with the rules suggested by Kaposi and Kov\'acs~\cite{kaposi_et_al:LIPIcs:2018:9190}, 
 we can form the \emph{syntax HIIT} by reading the specification as a HIIT signature.
 This gives us the initial wild CwF, but it is not a set and,
 by Hedberg's theorem~\cite{hedberg1998coherence}, does in particular not have decidable equality.
\end{example}

\begin{example}[{modelling an axiom; continuing \cref{ex:axmodel}}] \label{ex:axmodel2}
    The construction of the \emph{slice CwF} is broken for wild CwF's. 
    To be precise, the slice construction is already broken for wild semicategories before even considering identity morphisms, types, terms, or context extension.
    In a nutshell, when doing the slice construction of a (higher) category, a component at ``level'' $n$ needs components of the original category at level $n$ \emph{and} $n+1$: slice-objects need objects and morphisms, slice-morphisms need morphisms and composition, slice-composition needs composition and associativity, slice-associativity needs associativity and the coherence known as \emph{Mac Lane's pentagon} for bicategories.
    In \cref{ex:axmodel}, UIP makes this coherence automatic. Here, it is simply not possible to derive associativity.
\end{example}

The two issues that we have identified in \cref{ex:initial2,ex:axmodel2} are instances of the same underlying problem, namely data that is missing from the definition of a wild CwF.
In the first case, we need one datum that gives us $\idl(\id) = \idr(\id)$, and in the second case, we need a datum giving us the ``pentagon coherence''.
Unsurprisingly, na\"ively adding these to the definition of a wild CwF creates the need to add even more data, and so on.
The approach advocated in this paper, i.e.\ the use of $\infty$-categories, amounts to adding \emph{all} this data.

\begin{takeaway}
\noindent
\textbf{Section summary.}
When implementing models of type theory inside homotopy type theory using CwF's, one has to choose whether or not to impose a condition on the truncation levels of the involved types.
Both possibilities can be unsatisfactory.
\end{takeaway}

\section{Infinite Structures in Type Theory} \label{sec:infstrucs}

\cref{def:setCwF} defines a wild CwF to be a tuple with 26 components.
Even more extreme, \cref{def:wildCwF} defines a set-CwF to consist of 30 components.
While this is already tedious to express in type theory, it is entirely straightforward to express such a definition simply by listing the components.
Formally, a CwF is an element of a nested $\Sigma$-type or, if supported by the type theory, a record type (we view records as syntactic sugar for nested $\Sigma$-types).

However, what happens if we want to define a structure to consist of not only 30, but infinitely many components?
Often, this can be encoded in a finite way. To give a trivial example, an infinite sequence of elements of a type $A$ is simply given as the function type $\mathbb N \to A$ (or as a coinductive type of streams, although such coinductive definitions can be translated to homotopy type theory without coinduction \cite{paolo_nonwellfounded}).

A non-trivial example for a structure with infinitely many components is a \emph{semisimplicial type} \cite{UnivalentFoundationsProgram2013,Kraus:theroleofSST}, and it is a long-standing open problem whether a type capturing this structure can be defined in homotopy type theory.
Nevertheless, 
this structure is important for the approach to $\infty$-categories that we chose in this paper.
In this section, we first define what semisimplicial types are and then introduce the setting of \emph{two-level type theory (2LTT)}, an extension of HoTT which allows the treatment of certain structures with infinitely many components.

\subsection{Semisimplicial Types} \label{sec:sstypes}

A semisimplicial type of level $2$ is a tuple $(A_0, A_1, A_2)$ of types and type families as follows:
\begin{alignat}{3}
 & A_0 &&: &\; & \UU \label{eq:sst1} \\
 & A_1 &&: && A_0 \to A_0 \to \UU  \label{eq:sst2} \\
 & A_2 &&: && \{x_0 \, x_1 \, x_2 : A_0\} \to (A_1 \, x_0 \, x_1) \to (A_1 \, x_1 \, x_2) \to (A_1 \, x_0 \, x_2) \to \UU  \label{eq:sst3} 
\end{alignat}
We think of $A_0$ as a type of \emph{points} or \emph{vertexes}.
For two vertexes $x_0$ and $x_1$ we think of $A_1 \, x_0 \, x_1$ as a type of \emph{lines} between the vertexes.
For three vertexes and three lines forming a triangle, $A_2$ is the type of triangle fillers.

A semisimplicial type of level $3$ is a tuple $(A_0, A_1, A_2, A_3)$ with $A_0, A_1, A_2$ as above, and $A_3$ being a family of types that is indexed over four vertexes, six lines, and four triangle fillers which form the boundary of a tetrahedron.
Although harder to imagine geometrically, it is clear on an intuitive level how the next step would look like: a semisimplicial type of level $4$ adds a component $A_4$, which is a type family indexed over the boundary of a $4$-dimensional tetrahedron, consisting of five vertexes, ten lines, ten triangle fillers, and five tetrahedron fillers.

The mentioned long-standing open problem of homotopy type theory asks whether it is possible to \emph{define} semisimplicial types internally, i.e.\ whether it is possible to write down a type family $F : \mathbb N \to \UU_1$ (where $\UU_1$ is a universe that contains $\UU$) such that $F(n)$ encodes the type of tuples $(A_0, \ldots, A_n)$.
The problem was first discussed between Voevodsky, Lumsdaine, and others during the Univalent Foundations special year program at the IAS in Princeton 2012/13 and has since then acquired significant attention. Proposals for partial solutions, or extensions of type theory where solutions are possible, have been suggested by Voevodsky~\cite{voe_hts}, Herbelin~\cite{herbelin_semisimpl}, Part and Luo~\cite{DBLP:journals/corr/PartL15}, Altenkirch, Capriotti, and the current author~\cite{altCapKra_twoLevels}, and others.

Meta-theoretically, a semisimplicial type represents a certain contravariant functor from the category $\Delta_+$ to the category of (small) types.
This works as follows.
For a universe $\UU$, the category of types in $\UU$ has types as objects and functions as morphisms.
The categorical laws hold judgmentally/definitionally due to the $\eta$-law for $\Pi$-types.
Overloading notation, we denote this category by $\UU$.
The category $\Delta_+$ has natural numbers as objects. For a number $k$, we write $[k]$ for the finite set of numbers $\{0, 1, \ldots, k\}$.
Morphisms of $\Delta_+$ are strictly monotone functions,
\begin{equation}
 \Delta_+(i,j) \defeq \{\, f : [i] \to [j] \mid f \text{  strictly increasing}\, \}.
\end{equation}
We denote the full subcategory of $\Delta_+$ with objects $[0]$, $[1]$, $[2]$ by $\Delta_+^{\leq 2}$.
Drawing only generating morphisms, this category can be depicted as
\begin{equation} \label{eq:delta2}
 \begin{tikzpicture}[x=1.5cm,y=2cm,baseline=(current bounding box.center)]
  \node (V) at (0,0) {$[0]$};
  \node (E) at (1,0) {$[1]$};
  \node (F) at (2,0) {$[2]$};
  \draw[<-, transform canvas={yshift=+0.4ex}] (E) to node[above] {} (V);
  \draw[<-, transform canvas={yshift=-0.4ex}] (E) to node[below] {} (V);
  \draw[<-, transform canvas={yshift=+0.8ex}] (F) to node[above] {} (E);
  \draw[<-, transform canvas={yshift=+0ex}] (F) to node[] {} (E);
  \draw[<-, transform canvas={yshift=-0.8ex}] (F) to node[below] {} (E);
 \end{tikzpicture}
\end{equation}
Given a semisimplicial type of level $2$ as above, the functor $\left(\Delta_+^{\leq 2}\right)^{\mathsf{op}} \functor \UU$ that it represents is given by
\begin{equation} \label{eq:sst-as-functor}
 \begin{tikzpicture}[x=1.5cm,y=2cm,baseline=(current bounding box.center)]
  \node (V) at (0,0) [] {$A_0$};
  \node (E) at (2.5,0) [] {$\Sigma(x_0 \, x_1 : A_0). A_1 \, x_0 \, x_1$};
  \node (F) at (5.9,0) [align=left] 
    {$\Sigma(x_0 \, x_1 \, x_2 : A_0).$ \\ 
     $\Sigma (x_{01} : A_1 \, x_0 \, x_1).$ \\
     $\Sigma (x_{12} : A_1 \, x_1 \, x_2).$ \\
     $\Sigma (x_{02} : A_1 \, x_0 \, x_2).$ \\
     $\phantom{\Sigma (} A_2 \, x_{01} \, x_{12} \, x_{02}$
    };
  \draw[->, transform canvas={yshift=+0.4ex}, shorten >=0.2cm,shorten <=.2cm] (E) to node[above] {} (V);
  \draw[->, transform canvas={yshift=-0.4ex}, shorten >=0.2cm,shorten <=.2cm] (E) to node[below] {} (V);
  \draw[->, transform canvas={yshift=+0.8ex}, shorten >=0.2cm,shorten <=.2cm] (F) to node[above] {} (E);
  \draw[->, transform canvas={yshift=+0ex}, shorten >=0.2cm,shorten <=.2cm] (F) to node[] {} (E);
  \draw[->, transform canvas={yshift=-0.8ex}, shorten >=0.2cm,shorten <=.2cm] (F) to node[below] {} (E);
 \end{tikzpicture}
\end{equation}
The function between the types are simply projections, ensuring that the functor laws hold definitionally (thanks to $\eta$ for $\Sigma$-types).

Given a semisimplicial type $A$, we can consider an \emph{$A$-indexed family} of semisimplicial types.
With the view of a simplicial type as a functor, we can also call this a semisimplicial type $E$ \emph{over} $A$.
Concretely, if we have $A \equiv (A_0, A_1, A_2)$ as above, a semisimplicial type $E$ of level 2 over $A$ consists of:

\begin{alignat}{3}
& E_0 &&: &\; & A_0 \to \UU \label{eq:E-over-A-1} \\
& E_1 &&: && \{x_0 \, x_1 : A_0\} \to (x_{01} : A_1 \, x_0 \, x_1) \to (x_0' : E_0 \, x_0) \to (x_1' : E_1 \, x_1) \to \UU \label{eq:E-over-A-2}\\
& E_2 &&: && \{x_0 \, x_1 \, x_2 : A_0\} \to \{x_{01} : A_1 \, x_0 \, x_1\} \to \{x_{12} : A_1 \, x_1 \, x_2\} \to \{x_{02} : A_1 \, x_0 \, x_2\} \to \notag \\
&&&&& (x_{012} : A_2 \, x_{01} \, x_{12} \, x_{02}) \to \{x_0' : E_0 \, x_0\} \to \{x_1' : E_0 \, x_1\} \to \{x_2' : E_0 \, x_2\} \to \label{eq:E-over-A-3} \\
&&&&& 
(E_1 \, x_{01} \, x_0' \, x_1') \to 
(E_1 \, x_{12} \, x_1' \, x_2') \to 
(E_1 \, x_{02} \, x_0' \, x_2') \to \UU \notag 
\end{alignat}
Analogously to \eqref{eq:delta2}, the tuple $(A_0,A_1,A_2,E_0,E_1,E_2)$ can be viewed as a type-valued presheaf on the category
\begin{equation}
\begin{tikzpicture}[x=1.5cm,y=2cm,baseline=(current bounding box.center)]
\node (V) at (0,0) {$[0]$};
\node (E) at (1,0) {$[1]$};
\node (F) at (2,0) {$[2]$};
\draw[<-, transform canvas={yshift=+0.4ex}] (E) to node[above] {} (V);
\draw[<-, transform canvas={yshift=-0.4ex}] (E) to node[below] {} (V);
\draw[<-, transform canvas={yshift=+0.8ex}] (F) to node[above] {} (E);
\draw[<-, transform canvas={yshift=+0ex}] (F) to node[] {} (E);
\draw[<-, transform canvas={yshift=-0.8ex}] (F) to node[below] {} (E);

\node (V2) at (0,.7) {$[0']$};
\node (E2) at (1,.7) {$[1']$};
\node (F2) at (2,.7) {$[2']$};
\draw[<-, transform canvas={yshift=+0.4ex}] (E2) to node[above] {} (V2);
\draw[<-, transform canvas={yshift=-0.4ex}] (E2) to node[below] {} (V2);
\draw[<-, transform canvas={yshift=+0.8ex}] (F2) to node[above] {} (E2);
\draw[<-, transform canvas={yshift=+0ex}] (F2) to node[] {} (E2);
\draw[<-, transform canvas={yshift=-0.8ex}] (F2) to node[below] {} (E2);

\draw[<-] (V2) to node {} (V);
\draw[<-] (E2) to node {} (E);
\draw[<-] (F2) to node {} (F);
\end{tikzpicture}
\end{equation}

This presentation of type-theoretic structure as type families described by a \emph{one-way} or \emph{inverse} category is of course well-known since at least Makkai's FOLDS \cite{makkai:folds}.
For the theory of model categories, a similar concept has already been studied by Reedy~\cite{reedy:reedy}.
The theory of \emph{Reedy fibrant diagrams} in type theory have been developed by Shulman~\cite{shulman_inversediagrams}.
For the specific case of $\deltop$ (semisimplicial types), 
the current author has studied the relevance of the concept for homotopy type theory in earlier work%
~\cite{kraus_generaluniversalproperty,nicolai:thesis}.

\subsection{Two-Level Type Theory} \label{sec:two-level}

Voevodsky's \emph{Homotopy Type System (HTS)} \cite{voe_hts} 
is a version of HoTT which makes it possible to make certain meta-theoretical statements in the language of type theory.
\emph{Two-Level Type Theory (2LTT)}%
~\cite{paolo:thesis,altCapKra_twoLevels,ann-cap-kra:two-level}
develops the idea of HTS further and provides a more general framework.
2LTT allows us to work with infinite structures such as semisimplicial types.
The same can be done in other frameworks such as Shulman's inverse diagrams \cite{shulman_inversediagrams}, which would be completely suitable for the work of this paper as well; and the translation of the current paper to that setting is, in fact, straightforward.
For concreteness and to avoid leaving the type-theoretical setting, we nevertheless choose to work in 2LTT in the current paper.

The fundamental idea of HTS and 2LTT is very simple.
They contain a new type former, written $\steq$, which expresses \emph{strict equality}.%
\footnote{Similar as other authors~\cite{altCapKra_twoLevels,benedikt_HSIP} (but unlike~\cite{ann-cap-kra:two-level}), we leave the fibrant components unannotated and annotate the strict components instead. For simplicity, we do not distinguish between \emph{inner} and \emph{fibrant}; this can be achieved by assuming the axiom (T3) 
in the presentation by with Annenkov, Capriotti, Sattler, and the current author~\cite{ann-cap-kra:two-level}.}
Strict equality can be viewed as an ``internalised'' version of judgmental equality.
In fact, one can assume that $\steq$ satisfies equality reflection and thus $\steq$ and $\equiv$ coincide; this is an axiom which is part of HTS, and it has been shown that this assumption can be made without destroying the good property [conservativity] of 2LTT \cite{ann-cap-kra:two-level}.

However, it is not as simple as adding $\steq$ with all the usual rules for equality.
This would make $\steq$ and $=$ coincide, which is undesirable and incompatible with HoTT.
To resolve this issue, only some types of the theory will come with the equality type $=$, and these types are called \emph{fibrant}.
In other words, there are two ``layers'', or ``levels'', of types: fibrant types and not-necessarily-fibrant a.k.a.\ \emph{strict} types.

Fibrant types are the ones which are of interest and which correspond to types in HoTT.
Non-fibrant types can be thought of as auxiliary components which make the language more expressive, but which are of little interest themselves since they might not even exist in HoTT.
If $A$ is a fibrant type and $x,y : A$ elements, then $(x=y)$ is a fibrant type while $(x \steq y)$ is (in general) not: the first is a perfectly fine type of HoTT, while the second is something that cannot be expressed in HoTT.
In contrast, $\Sigma (x' : A). x' \steq x$ \emph{is} fibrant, since it is strictly isomorphic to the (fibrant) unit type; i.e.\ this exists in HoTT up to strict isomorphism.
Similar as for equality, 2LTT has two types of natural numbers: First, the standard type $\N$ of fibrant natural numbers which one works with in type theory; and second, the type $\sN$ of \emph{strict} natural numbers. Similarly, we have fibrant finite types $\Fin \, n$ (for $n : \N$) and strict finite types $\Fin^s \, n$ (for $n : \sN$).
HTS assumes that $\N$ and $\sN$ coincide, which is justified by the model in simplicial sets \cite{kap-lum:simplicial-model}.

2LTT allows us to make the description of \cref{sec:sstypes} of a semisimplicial type as a functor precise.
We briefly recall the construction of~\cite{ann-cap-kra:two-level}.
Using strict equality, we can define all the standard categorical concepts.
\begin{definition}[strict categories and their theory]
  A \emph{strict category} is a tuple $(\ob,\hom,\mc,\ass,\id,\idl,\idr)$ of standard categorical structure, where the identity laws ($\idl$, $\idr$, $\ass$) are formulated using strict equality ($\steq$).
  Strict functors, strict natural transformations, and so on are analogously defined.
\end{definition}
Given a universe $\UU$ of fibrant types, the strict category of types in $\UU$ and functions is denoted by $\UU$ as well.
Similarly, we can construct the category $\deltop$ using strict components (objects are $\sN$).
We can then consider the (non-fibrant) type of functors $\deltop \functor \UU$, and define the (non-fibrant) structure of being \emph{Reedy fibrant}.
We can then define the type of Reedy fibrant diagrams over $\deltop$, i.e.\ pairs of a functor and a proof that it is Reedy fibrant, for which we write $\deltop \rffunctor \UU$.

Given $A, E : \deltop \rffunctor \UU$ and a natural transformation $\eta : E \to A$, one can define what it means for $\eta$ to be a \emph{Reedy fibration}, which means that $E$ is a semisimplicial type over $A$ as explained in the previous subsection.
It is standard to denote fibrations by $\eta : E \twoheadrightarrow A$ (although attention is required when there are multiple different notions of fibrations).
We refer to \cite{ann-cap-kra:two-level} for the precise definition.

The type $\deltop \rffunctor \UU$ encodes the intuitive idea of a type of ``infinite tuples'' $(A_0, A_1, A_2, \ldots)$.
This type is (in general) not fibrant.
Given a number $n : \sN$, we can take the full subcategory of $\deltop$ of objects $\{[0], [1], \ldots, [n]\}$.
The corresponding type $\finitedeltop n \rffunctor \UU$ 
is fibrant and can be understood as
the nested $\Sigma$-type of finite tuples $(A_0, A_1, \ldots, A_n)$, i.e.\ as a semisimplicial type of level $n$.
It is possible to equip 2LTT with an axiom which ensures that the type $\deltop \rffunctor \UU$ is fibrant as well, e.g.\ by assuming axiom (A2) of \cite{ann-cap-kra:two-level}, and the models of 2LTT discussed there all satisfy this axiom.
The original suggestion by Voevodsky~\cite{voe_hts}, corresponding to the strongest form of (A1) in~\cite{ann-cap-kra:two-level}, has the same effect, but unnecessarily invalidates the last of the models discussed in \cite[Chp 2.4]{ann-cap-kra:two-level}.

More generally, everything we have said for the type $\deltop \rffunctor \UU$ is also the case for the type of semisimplicial types over a given semisimplicial type $A$.

Our presentation of 2LTT above is significantly simplified but sufficient for the current paper.
For details, we refer to \cite{ann-cap-kra:two-level}.
Apart from axiom (A2), we assume the axioms (M1, M2, T1, T2, T3), although this is only to simplify the presentation; what we do could be done without the latter five axioms.

\begin{takeaway}
\noindent
\textbf{Section summary.}
By \emph{infinite structure}, we mean what can intuitively be described as a record type with infinitely many components, or an infinitely nested $\Sigma$-type.
Semisimplicial types are an example of such a structure where it is unclear whether it can be encoded in homotopy type theory.

In 2LTT, we can define the type $\deltop \rffunctor \UU$, elements of which correspond to the intuitive idea of an infinite tuple $(A_0, A_1, A_2, \ldots)$ describing a semisimplicial type.
There is a reasonable axiom which makes the type $\deltop \rffunctor \UU$ fibrant, i.e.\ lets it behave like a type in homotopy type theory.
\end{takeaway}

\section{Higher Dimensional Categories and Internal $\boldsymbol{\infty}$-CwF's}
\label{sec:highercwf}

We now have all the tools to define \emph{$\infty$-categories with families}.

\subsection{Contexts and substitutions, first part: semi-Segal types}
\label{subsec:highersemicats}

\emph{Segal spaces} (Rezk~\cite{rezk2001model}), also called \emph{Rezk spaces}, are a model for higher categories.
Translating Rezk's Segal condition to type theory is straightforward and has been presented previously by Capriotti~\cite{paolo:thesis} and others, but note that type theory only allows us to talk about \emph{semi}simplicial types.
We briefly sketch the definition as presented in~\cite{ann-cap-kra:two-level}.

On the strict level, and for $n : \sN$, we can define the semisimplicial set $\simplex n$ as the functor represented by $[n]$. 
This is the $n$-dimensional tetrahedron.
A trivial application of the Yoneda lemma shows that, given a semisimplicial type $A$, the type of natural transformations $\simplex n \to A$ is strictly isomorphic to the total space of $A_n$ (a collection of $n+1$ points in $A_0$, $\binom {n+1} 2$ lines, and so on).
Given another number $k : \sN$, we can define the \emph{horn} $\horn n k$ as a semisimplicial set; $\horn n k$ can be obtained from $\simplex n$ by removing the single cell at level $n$, and the cell at level $(n-1)$ which does not depend on the $k$-th point.

There is a canonical inclusion $\horn n k \hookrightarrow \simplex n$.
A semisimplicial type $A$ satisfies the Segal condition if, for any given $n$ and $0 < k < n$, a given $\eta : \horn n k \to A$ can uniquely be extended to a natural transformation $\simplex n \to A$. Here, ``uniquely extending'' means that the type of extensions is contractible. $A$ is also called \emph{local} with respect to $\eta$.
\begin{definition}
    A \emph{semi-Segal type} is a semisimplicial type which satisfies the Segal condition.
\end{definition}    
    
Capriotti and the current author~\cite{capKra_semisegal} unfold the Segal condition explicitly for small $n,k$ (and actually fix $k \equiv 1$, which is sufficient).
Given a semisimplicial type $(A_0, A_1, A_2)$ of level 2, the Segal condition can be expressed as:
\begin{equation}
\begin{aligned}
 h_2 : \; &\{x_0 \, x_1 \, x_2: A_0\} \to (x_{01} : A_1 \, x_0 \, x_1)
\to (x_{12} : A_1 \, x_1 \, x_2)
\to \\
& \iscontr\left(\Sigma(x_{02} : A_1 \, x_0 \, x_2). A_2 \, x_{01}\, x_{12}\, x_{02}\right)
\end{aligned}
\end{equation}
We also call this a \emph{horn-filling condition} or the condition that the horn $x_0 \xrightarrow{x_{01}} x_1 \xrightarrow{x_{12}} x_2$ has a \emph{contractible type of fillers}.

In~\cite{capKra_semisegal}, an explicit translation between horn-filling conditions and the structure one expects of (higher) semicategories is constructed.
It not hard to see that having the pair $(A_2,h_2)$ is equivalent to having a composition operator $(\blank \fc \blank) : \, A_1 \, x_0 \, x_1 \to A_1 \, x_1 \, x_2 \to A_1 \, x_0 \, x_2$:
having an element of a type (here $A_1 \, x_0 \, x_2$) is equivalent to having a family over the type (here $A_2 \, x_{01} \, x_{12} : A_1 \, x_0 \, x_2 \to \UU$) and a prove that it is inhabited at exactly one point (here $h_2 \, x_{01} \, x_{12}$).
On the next level, we require that any map $\horn 3 1 \to A$ can uniquely be extended to $\simplex 3 \to A$, and this (together with $A_3$) is equivalent to stating that $\blank \fc \blank$ is associative.
The next level corresponds to the pentagon coherence of associativity familiar from the definition of a bicategory~\cite{capKra_semisegal}.
This explains why semi-Segal types can be seen as $\infty$-semicategories.

The generalisation from semisimplicial types to maps between semisimplicial types is canonical:

\noindent
\begin{minipage}[t]{\textwidth -4.4cm}
\begin{definition}[inner fibration] \label{def:innerfibration}
 An \emph{inner fibration}, or a \emph{semi-Segal type over a semisimplicial type}, is a Reedy fibration $\eta : E \twoheadrightarrow A$ such that, for all $n$, $0 < k < n$, and squares of the form
 shown on the right, the type of fillers (i.e.\ the type of the dashed map) is contractible.
\end{definition}
\end{minipage}
\hfill
\begin{minipage}[t]{4cm}
    
    \vspace*{-.1cm}
    
  \noindent  
 \begin{equation}
\begin{tikzpicture}[x= 2cm , y = 1.3cm, baseline=(current bounding box.center)]
\node (H) at (0,1) {$\horn n k$};
\node (S) at (0,0) {$\simplex n$};
\node (E) at (1,1) {$E$};
\node (A) at (1,0) {$A$};

\draw[right hook->] (H) to node[left] {} (S);
\draw[->] (H) to node[left] {} (E);
\draw[->] (S) to node[above] {} (A);
\draw[->>] (E) to node[right] {$\eta$} (A);

\draw[->,dashed] (S) to node[right] {} (E);
\end{tikzpicture}
\end{equation}
\end{minipage}

By definition, a semisimplicial type $A$ satisfies the Segal condition exactly if $A \twoheadrightarrow 1$ is an inner fibration.

\subsection{Contexts and substitutions, second part: identities}
\label{subsec:identities}

Since a semi-Segal type can be taken as the definition of an $\infty$-semicategory, the natural next question is how we can extend this definition to capture $\infty$-categories.
To do this, we need to add a property which ensures that the $\infty$-semicategory has identities.
We want this property to satisfy the following conditions:
\begin{enumerate}
    \item Of course, the property should give us (i.e.\ allow us to construct) the ``na\"ive'' identities that are shown in \cref{fig:Ambrus-changed-reordered}: For every object, we want an endomorphism $\id$ such that $\id \mc \sigma = \sigma$ and $\sigma \mc \id = \sigma$.
    (Of course, simply adding this ``na\"ive'' structure itself without coherences would be unsatisfactory.)
    \item The property should be a proposition. In other words, we want an $\infty$-semicategory to be an $\infty$-category in at most one way. This is important because we do not only want to talk about a single $\infty$-category in isolation, but about a \emph{type} of $\infty$-categories.
    \item From the perspective of HoTT, it is natural to want \emph{univalent} (or \emph{saturated}) $\infty$-categories,
    meaning that identities/equalities and isomorphisms coincide (analogous to univalent 1-categories~\cite{ahrens_rezk}).
    One would expect (and we will show this later in this paper) that the standard model is indeed a \emph{univalent} $\infty$-CwF.
    However, not all examples that we want can possibly be univalent: In a type theory, contexts can be isomorphic in non-trivial ways, meaning that the type $\Con$ in a univalent $\infty$-CwF cannot be a set. This would rule out the ``syntax'' (\cref{ex:initial}). Phrased differently, the initial model of a non-trivial type theory would not be on the level of sets
    (and will in particular not have decidable equality for contexts), which contradicts the expectation by Altenkirch~\cite{alt_nicolaitalk} and others.%
    \footnote{In a private conversation with the current author, Eric Finster has pointed out that type-checking algorithms do not rely on context equality being decidable. In other words, one could potentially drop the condition that contexts of the syntactic model form a set, while types and terms certainly should. This may make a univalent syntactic model possible.\label{finster-footnote}}
    Therefore, we want a more general definition of ``being an $\infty$-category'' (a.k.a.\ ``having identities'') which allows us to study both univalent and non-univalent variants.
\end{enumerate}

The question of adding identities to $\infty$-semicategories
has been studied outside the field of type theory long before HoTT was researched.
It is known~\cite{rourke1971delta} that a semisimplicial \emph{set} with (not necessarily unique) fillers for all (not necessarily inner) horns can be extended to a simplicial set. This is however non-constructive and relies on choice.
Our situation is closer to semisimplicial \emph{spaces}
with an inner horn filling condition, for which Lurie~\cite{lurie:higher-algebra} and Harpaz~\cite{harpaz2015quasi} have suggested a simple condition that expresses the presence of suitable identities.
The translation of their idea into type theory leads to the \emph{complete semi-Segal types} by Capriotti~\cite{paolo:thesis} and others.
Complete semi-Segal types can be taken as the definition of type-theoretic univalent $\infty$-categories, but the univalence condition is built into the definition of identities.

A \emph{direct replacement} of the full simplex category
has been suggested by Sattler and the current author~\cite{kraus-sattler:spacediagrams}, which gives rise to (a replacement of) simplicial types and thereby (not necessarily univalent) $\infty$-categories.
Another direct replacement was given by Kock~\cite{kock2006weak}.
These approaches both work by adding an infinite tower of data.

The current paper gives a new and different definition of identities in $\infty$-semicategories.
The idea is based on what is known as \emph{dunce's hat} in topology~\cite{zeeman1963dunce}, which is the simplest example of a space that is contractible but not collapsible (in type-theoretic terms: a space that is contractible but not of the form $\Sigma(a:A). a = a_0$).
Concretely, we define identities to be \emph{idempotent equivalences}.
This is both minimalistic and, as we prove in this section, well-behaved; it satisfies the three properties listed above.
This definition corresponds to a characterisation of identities that was independently suggested by Lai~\cite{luhanglai}, who considered $(\infty,1)$-categories in a slightly different version of type theory.

This author expects that the definition of identities via idempotent equivalences is equivalent to both of the mentioned ``elaborate'' definitions \cite{kock2006weak,kraus-sattler:spacediagrams}. 
If we add univalence, the equivalence with the complete semi-Segal types is obvious.

Our definition of an identity does not need the whole infinite structure of a semi-Segal type, but only the first four levels $(A_0, A_1, A_2, A_3)$.
Using the translation explained in \cref{subsec:highersemicats}, we phrase this subsection in the (maybe more familiar) language of a semicategory $(\ob, \hom, \mc, \ass)$, without any requirement of set-truncation.
This allows us to present the development in a very elementary way, without the need to allude to infinite structures.
This development has been completely formalised in Agda (see \cref{footnote:agda-links} on page \pageref{footnote:agda-links}).

For the remainder of this section, we assume that we are given a wild semicategory 
\begin{definition}[identities via idempotent equivalences]
  In a wild semicategory $(\ob, \hom, \mc, \ass)$,
  a morphism $e : \hom (x,y)$ is an \emph{equivalence} (\emph{neutral morphism} in \cite{capKra_semisegal}), written $\iseqv(e)$,
  if both composition operations
 \begin{align}
  & (\blank  \mc e) : \Pi\{z:\ob\}. \hom(y,z) \to \hom(x,z) \\
  & (e \mc \blank)  : \Pi\{w:\ob\}. \hom(w,x) \to \hom(w,y) 
 \end{align}
 are equivalences of types.
 A morphism $f : \hom(x,x)$ is \emph{idempotent} if $f \mc f = f$.
 A morphism $i : \hom(x,x)$ is a \emph{good identity} if it is an equivalence and idempotent.
 The wild semicategory $C \equiv (\ob,\hom,\mc,\ass)$ has a 
 \emph{good identity structure}
 if there is an identity for every object,
 \begin{equation}
  \mathsf{hasGoodIdStruc}(C) :\equiv \Pi(x:\ob). \Sigma(i: \hom(x,x)). \iseqv(i) \times (i \mc i = i).
 \end{equation}
\end{definition}

\begin{remark}
 Note that ``being an equivalence'' is a propositional property, while ``being idempotent'' and ``being a good identity'' are data.
 The non-trivial result of this section is that ``having a good identity structure'' is a propositional property.
\end{remark}

\begin{lemma} \label{lem:id-characterisation}
 A morphism $i : \hom(x,x)$ is a good identity if and only if,
 for all composable morphisms
 $\xrightarrow f \xrightarrow i \xrightarrow g$,
 we have $i \mc f = f$ and $g \mc i = g$.
\end{lemma}
\begin{proof}
 Let $i$ be an idempotent equivalence and $f : \hom(w,x)$.
 We have $i \mc i \mc f = i \mc i \mc f$ trivially, thus by idempotence (and implicitly associativity) also $i \mc i \mc f = i \mc f$, and by applying $(i \mc \blank)^{-1}$ to both sides $i \mc f = f$.
 The proof of $g \mc i = g$ is analogous.
 
 Conversely, if $i$ is such that both compositions are the identity, it is clearly an equivalence and idempotent.
\end{proof}

\begin{corollary} \label{cor:only-one-id}
 If $i_1, i_2 : \hom(x,x)$ are both good identities, then $i_1 = i_2$.
\end{corollary}

Note that \cref{lem:id-characterisation} does not state an equivalence of types, but only claims ``functions in both directions''.
We have seen in \cref{ex:initial2} that the ``na\"ive'' characterisation causes coherence issues.
In contrast, the following shows that the definition via idempotent equivalences is fully coherent:

\begin{theorem} \label{thm:id-struc-is-prop}
 For a given semicategory $\CC$, the type $\mathsf{hasGoodIdStruc}(C)$ is a proposition.
\end{theorem}

Before proving this theorem, we formulate several very simple auxiliary statements.
Given any equivalence $e : \hom(x,y)$, we define $I(e) : \hom(x,x)$ by
\begin{equation} \label{eq:I-def}
  I(e) \defeq (e \mc \blank )^{-1}(e).
\end{equation}
This is how identities (modulo the translation of \cref{subsec:highersemicats}) are constructed by Harpaz~\cite{harpaz2015quasi} and by Capriotti and the author~\cite{capKra_semisegal}, where sufficiently many equivalences are assumed as given.

\begin{lemma} \label{lem:I-is-id}
 For any equivalence $e : \hom(x,y)$, the morphism $I(e)$ is a good identity.
\end{lemma}
\begin{proof}
We can show idempotency directly by setting $f \defeq I(e)$ in the following calculation: 
 \begin{alignat}{12}
  &I(e) \mc f &\quad & = & \quad & (e \mc \blank)^{-1}\left((e \mc \blank) (I(e) \mc f)\right) \\
  &&& = && (e \mc \blank)^{-1}\left((e \mc I(e)) \mc f \right) \\
  &&& = && (e \mc \blank)^{-1}\left(e \mc f \right) \\
  &&& = && f
 \end{alignat}
 Furthermore, $I(e)$ is an equivalence since equivalences satisfy \emph{2-out-of-3} and $e \mc I(e) = e$.
\end{proof}

\begin{lemma} \label{lem:e-I-idem}
 Let $e : \hom(x,x)$ be an equivalence. The type witnessing that $e$ equals $I(e)$ is equivalent to the type witnessing that $e$ is idempotent,
 \begin{equation} \label{eq:e-Ie}
  (e = I(e)) \simeq (e \mc e = e)
 \end{equation}
\end{lemma}
\begin{proof}
 The equivalence is $\ap_{e \mc \blank}$ (i.e.\ applying $(e \mc \blank)$ on both sides of the equation).
\end{proof}

We are ready to show that, given a semicategory, there is at most one good identity structure.
\begin{proof}[Proof of \cref{thm:id-struc-is-prop}]
 Let us fix $x:\ob$; by function extensionality, it suffices to prove that $\Sigma(i:\hom(x,x)). \iseqv(i) \times (i \mc i = i)$ is propositional. Assume we are given an element $(i_0, p, q)$; then:
\begin{alignat}{12}
&&&&& \Sigma(i:\hom(x,x)). \iseqv(i) \times (i \mc i = i) \\
&\text{by \cref{lem:e-I-idem}} &\quad & \simeq & \quad & \Sigma(i:\hom(x,x)). \iseqv(i) \times (i = I(i))  \\
&\text{by \cref{cor:only-one-id,lem:I-is-id}} &\quad & \simeq & \quad & \Sigma(i:\hom(x,x)). \iseqv(i) \times (i = i_0)
\end{alignat}
This type is easily seen to be contractible, with $(i_0, p, \refl)$ as centre of contraction.
\end{proof}

Phrased in the language of a semi-Segal type $(A_0, A_1, A_2, \ldots)$, 
a morphism $e : A_1 \, x \, y$ is an equivalence if any horn of the form $1 \xleftarrow e 0 \xrightarrow {} 2$ or $0 \xrightarrow {} 2 \xleftarrow e 1$
has a contractible type of fillers.
A proof that the morphism $f : A_1 \, x \, x$ is idempotent is an element of $A_2 \, f \, f \, f$.

\begin{definition} \label{def:infcategory}
    An $\infty$-category is an $\infty$-semicategory with a good identity structure.
\end{definition}

Univalence for $\infty$-categories is an optional additional property which we discuss in \cref{sec:variations}.
From the next subsection on, we will drop the word \emph{good} and simply speak of \emph{identities} rather than \emph{good identities}.
There is no risk of confusion, especially since \cref{lem:id-characterisation} essentially shows that \emph{every} identity is good.

\begin{remark}
Of course, the above results show in particular that ``being a good identity'' is equivalent to ``being left- and right-neutral'' if $\hom$ is a family of sets.
Thus, the usual definition of (pre-)category as given by Ahrens, Kapulkin, and Shulman~\cite{ahrens_rezk} can equivalently be formulated using good identities instead of ``standard'' identities.
\end{remark}

\subsection{The empty context: a terminal object}

Modelling the empty context as a terminal object in the category of contexts is standard.

\begin{definition}[terminal object] \label{def:terminal}
 Given an $\infty$-category $(A_0, A_1, A_2, \ldots)$, an object $x : A_0$ is \emph{terminal} if, for all $y : A_0$, the type $A_1 \, y \, x$ is contractible.
\end{definition}

\subsection{Types: a presheaf} \label{subsec:types}

Following \cref{def:CwF}, types should be a presheaf on the category of contexts.
This author is aware of two differently looking (but of course in a suitable sense equivalent) ways to define what an \emph{$\infty$-presheaf} (or \emph{prestack}) is in this setting.
The first possibility is to:
\begin{enumerate}
    \item define the \emph{opposite category} $\Acat^\mathsf{op}$;
    \item define the $\infty$-category $\TT$ of types and functions;
    \item and define what an \emph{$\infty$-functor} between $\infty$-categories is. \label{item:functor}
\end{enumerate}
The second possibility, suggested and studied by Christian Sattler in unpublished work, is to consider \emph{right fibrations} over the category of contexts.
We will discuss this second approach and its advantages below in \cref{rem:left-fib-classifier}.
For now, we concentrate on the first possibility which is closer to the 1-categorical formulation discussed in \cref{sec:tt-in-tt}.

Let us start with the last point \eqref{item:functor}. A morphism in a functor category is a natural transformation, and $\infty$-categories are certain functors (equipped with structure).
Manually unfolded to type theory, a natural transformation between semisimplicial types of level 2, from $(A_0, A_1, A_2)$ to $(B_0, B_1, B_2)$ is a tuple $(F_0, F_1, F_2)$ where:

\begin{alignat}{3}
& F_0 &&: &\; & A_0 \to B_0 \\
& F_1 &&: && \{x_0 \, x_1 : A_0\} \to (A_1 \, x_0 \, x_1) \to (B_1 \, (F_0 \, x_0) \, (F_0 \, x_1)) \\
& F_2 &&: && \{x_0 \, x_1 \, x_2 : A_0\} \to \{x_{01} : A_1 \, x_0 \, x_1\} \to \{x_{12} : A_1 \, x_1 \, x_2\} \to \{x_{02} : A_1 \, x_0 \, x_2\} \to \notag \\
&&&&& (A_2 \, x_{01} \, x_{12} \, x_{02}) \to (B_2 \, (F_1 \, x_{01}) \, (F_1 \, x_{12}) \, (F_1 \, x_{02})
\end{alignat}

Being a map between the underlying semisimplicial types already ensures that $F$ preserves the compositionality stucture.
We can make this transparent on level 2
via the translation explained in \cref{sec:sstypes},
under which the last component $F_2$ translates to:
\begin{alignat}{3}
& F_2' &&: &\; & \{x_0 \, x_1 \, x_2 : A_0\} \to \{x_{01} : A_1 \, x_0 \, x_1\} \to \{x_{12} : A_1 \, x_1 \, x_2\} \to
\notag \\
&&&&& (F_1 \, x_{12}) \mc (F_1 \, x_{01}) = F_1 (x_{12} \mc x_{01})
\end{alignat}

$\infty$-categories come with a proof that the underlying functors are Reedy fibrant, but the natural transformation takes care of this automatically.
Thus, the only task that is left is to ensure that the functor $F$ preserves identities, i.e.\ if $i : A_1 \, x \, x$ is an idempotent equivalence, then $F_1 \, i$ has to be an idempotent equivalence as well.
Stating it in this way would not give a well-behaved notion of $\infty$-functor, since ``being idempotent'' is not a propositional property.
Fortunately, idempotence is already preserved automatically by $F_2$, and ``being an equivalence'' \emph{is} a propositional property.

\begin{definition}[$\infty$-functor] \label{def:ifunctor}
    An \emph{$\infty$-functor} between $\infty$-categories is a natural transformation $F \equiv (F_0, F_1, F_2, \ldots)$ which maps identities to equivalences, i.e.\ for which we have:
    \begin{equation}
     \mathsf{id}\mbox{-}\mathsf{preserving}(F) : (x : A_0) \to (i : A_1 \, x \, x) \to \isid(i) \to \iseqv(F_1 \, i)
    \end{equation}
\end{definition}

\begin{remark}
  A natural transformation between semi-Segal types does not automatically preserve equivalences or even identities. Let $\top$ be the terminal semi-Segal type, which is the unit type $\unit$ on every level.
  Let $A$ be the semi-Segal type defined by $A_0 \defeq \unit$, $A_1 \, \blank \defeq (\bool \to \bool)$, $A_2 \, f \, g \, h \defeq (g \fc f = h)$, and $A_{k+3}$ being constantly $\unit$. There are three maps $\top \to A$, each of them even an inner fibration, but only a single of them preserves identities. 
\end{remark}

Given an $\infty$-category $A \equiv (A_0, A_1, A_2, \ldots)$, we need to define the opposite category $A^{\mathsf{op}} \equiv (A_0^{\mathsf{op}}, A_1^{\mathsf{op}}, A_2^{\mathsf{op}}, \ldots)$.
There is really only one possible construction: of course we want $A_0^{\mathsf{op}} \defeq A_0$ and $A_1^{\mathsf{op}} \, x \, y \defeq A \, y \, x$, and after this, everything is determined; e.g.\ we have $A_2^{\mathsf{op}} \, f \, g \, h \defeq A_2 \, g \, f \, h$, since nothing else would type-check.
The concrete combinatorial description for the representation as functors is identical to the one for simplicial sets given by Lurie~\cite[Sec~1.2.1]{lurie:higher-topoi}.

Finally, we need the $\infty$-category $\TT$ of types and functions,
starting with $\TT_0 \equiv \UU$, $\TT_1 \, X \, Y \defeq (X \to Y)$, and $\TT_2 \, f \, g \, h  \equiv (g \fc f = h)$.
The complete construction has been given in~\cite{ann-cap-kra:two-level}, which we briefly sketch here.
Given any \emph{strict} category $C$, such as the category of types and functions, one can define $\nerve(C) : \deltop \functor \UU$ via the usual nerve construction which defines $\nerve(C)_n$ to be sequences $X_0 \to X_1 \to \ldots \to X_n$.
This functor is not Reedy fibrant, but via a general Reedy fibrant replacement construction \cite[Sec.\ 4]{ann-cap-kra:two-level}, one can find a levelwise equivalent and Reedy fibrant functor $R(\nerve(C))$.
The identity structure is obvious and simply comes from $X \xrightarrow{\mathsf{id}} X$.

\begin{definition}[$\infty$-category of small types]\label{def:TT}
    The $\infty$-category $\TT$ is given as $R(\nerve(\UU))$.
\end{definition}

\subsection{Terms: diagrams over a category of elements} \label{subsec:terms}

Recall from \cref{def:CwF} that terms are in 1-CwF's modelled by a functor $\Tm : \int_{\CCop} \Ty \to \Set$.
Our treatment of types has already covered several of the components that are needed in order to translate this to the $\infty$-categorical setting.
The remaining missing part is the construction of the \emph{category of elements}, which is what we explain in this subsection.

Let $\Acat$ be an $\infty$-category and $F : \Acat \to \TT$ be an $\infty$-functor.
As one may expect, $\int_{\Acat} F$ can be viewed as an $\infty$-category \emph{over} $\Acat$: Given an $n$-simplex in $\Acat$ consisting of vertexes $x_0, x_1, \ldots$, lines $x_{01}, \ldots$, and so on, then we can build an $n$-simplex in $\int_{\Acat} F$ if we have something in $F_0 \, x_0$, in $F_0 \, x_1$, in $F_1 \, x_{01}$, and so on, where all the new data has to ``match''. 
Our goal is to define $\int_{\Acat} F$ on all levels.

We first solve the problem for the case that $A$ is a strict category and $F : A \to \UU$ a strict functor
with the following ad-hoc construction, which has been (partially) discussed in \cite{ann-cap-kra:two-level}.
We define $\nervep(A,F)$ to be the functor $\deltop \functor \UU$ given by the nerve together with a single point: An element of $\nervep(A,F)_n$ is a pair $(s, p)$ of a chain $s : x_0 \to x_1 \to \ldots \to x_n$ in $A$ and a point $p : F \, x_0$.
There is a canonical natural transformation $\eta : \nervep(A,F) \to \nerve(A)$ which, on each level, forgets the point.
Using the mentioned Reedy fibrant replacement \cite[Sec\ 4]{ann-cap-kra:two-level} twice, we get a Reedy fibration 
between Reedy fibrant diagrams.

For the terminal (``universal'') case, where $A$ is $\UU$ and $F$ the identity functor, we denote the constructed Reedy fibration 
by $\TTp \twoheadrightarrow \TT$. 
We call $\TTp$ the \emph{$\infty$-category of pointed types}.
For the first few levels, its ``over $\TT$'' representation (cf.\ \ref{eq:E-over-A-1}--\ref{eq:E-over-A-3}) is the following, where we annotate arguments with their types for readability
and $\mathsf{happly}$ is the function given in \cite[Eq.\ 2.9.2]{hott-book}:

\begin{alignat}{3}
& \TTp_0 \, (X : \UU) &\quad & \defeq &\quad & X \\
& \TTp_1 \, (f : X \to Y) \, (x : X) \, (y : Y) && \defeq && (f \, x = y) \\
& \TTp_2 \, (\alpha : g \fc f = h) \, (e_0 : f \, x = y) \notag \\ 
& \phantom{\TTp_2 \,} (e_1 : g \, y = z) \, (e_2 : h \, x = z) && \defeq && e_2 = (\mathsf{happly} \, \alpha \, x) \ct (\ap_g \, e_0) \ct e_1
\end{alignat}

We can now come back to the general case of an $\infty$-category $\Acat$ and an $\infty$-functor $F : \Acat \to \TT$.
Recall (from \cite{ann-cap-kra:two-level}) that Reedy fibrations are closed under pullback.

\vspace*{+.0cm}

\noindent
\begin{minipage}[t]{\textwidth -4.2cm}
    
\begin{definition}[$\infty$-category of elements] \label{def:catofels}
    For an $\infty$-category $\Acat$ and a functor $F : \Acat \to \TT$, the \emph{$\infty$-category of elements of $F$}, written $\int_{\Acat} F$, is defined to be the strict pullback of $F$ along the Reedy fibration $\TTp \twoheadrightarrow \TT$, as shown on the right.
\end{definition}
\end{minipage}
\hspace{\fill}
\begin{minipage}[t]{4cm}
    \vspace*{-.5cm}
    \begin{equation} \label{eq:el-def}
\begin{tikzpicture}[x= 1.5cm , y = -1.25cm, baseline=(current bounding box.center)]
\node (El) at (0,0) {$\int_\Acat F$};
\node (A) at (0,1) {$\Acat$};
\node (PT) at (1,0) {$\TTp$};
\node (T) at (1,1) {$\TT$};

\draw +(.15,.3) -- +(.3,.3)  -- +(.3,.15);

\draw[->>,dashed] (El) to node[left] {$\pi_1$} (A);
\draw[->,dashed] (El) to node[left] {} (PT);
\draw[->] (A) to node[above] {$F$} (T);
\draw[->>] (PT) to node[above] {} (T);
\end{tikzpicture}
\end{equation}
\end{minipage}

\vspace*{.2cm}

It is standard that pullbacks are closed under [left] fibrations (given a lifting problem for $\int_{\Acat} F \twoheadrightarrow \Acat$, we extend it to a lifting problem for $\TTp \twoheadrightarrow \TT$ and use the pullback property).
The composition of a left and an inner fibration is an inner fibration by a similar argument. 
Therefore, $\int_{\Acat}F$ is an $\infty$-category.

\begin{remark}[left and right fibrations] \label{rem:left-fib-classifier}
    Let $\eta : E \twoheadrightarrow A$ be a Reedy fibration between semisimplicial types.
    $\eta$ is a \emph{left fibration} if it fulfils the condition of \cref{def:innerfibration} for $0 \leq k < n$ (that means that all left fibrations are inner fibrations, but only some inner fibration are left fibrations).
    Analogously, $\eta$ is a \emph{right fibration} if the condition of \cref{def:innerfibration} holds for for $0 < k \leq n$. 
    
    One can show that the map $\eta : \nervep(A,F) \to \nerve(A)$ constructed above, and in particular $\TTp \twoheadrightarrow \TT$, is a left fibration.
    One can further show that $\TTp \twoheadrightarrow \TT$ is a \emph{homotopical left fibration classifier}:
    For any left fibration $E \twoheadrightarrow A$, the type of 
    tuples $(g,h,q)$ with $g : A \to \TT$, $h : E \to \TTp$, and $q$ witnessing that the resulting square commutes up to homotopy and is a homotopy pullback, is contractible.
    
    This implies that left fibrations over $A$ are in a suitable sense equivalent to $\infty$-functors $A \to \TT$.
    Thus, $\infty$-presheaves on $A$ correspond to right fibrations over $A$.
    If $\CC$ is the $\infty$-category of contexts, we can define types to be given as a right fibration $\mathsf{Ty} \twoheadrightarrow \CC$.
    Since $\mathsf{Ty}$ already is (the opposite of) the $\infty$-category of elements, terms are then simply given by a second right fibration $\mathsf{Tm} \twoheadrightarrow \mathsf{Ty}$.
    This formulation is due to Christian Sattler.
    
    Note that the strict pullback \eqref{eq:el-def} is also a homotopy pullback since the vertical maps are Reedy fibrations.
\end{remark}

\subsection{Context extension}

Context extension is arguably the most involved component of a CwF in the 1-categorical setting.
For what it is worth, context extension contributes the largest number of components in the type theoretic representation in \cref{fig:Ambrus-changed-reordered}.
Even so, when generalising from CwF's to $\infty$-CwF's, context extension greatly benefits from the formulation of representabiliy via initiality in a category of elements (see \cref{subsec:GAT}).
Even for $\infty$-categories, representability can be stated referring only to the lowest levels of the category.
This means that context extension can be stated using only finitely many components.
Still, the presentation of \cref{fig:Ambrus-changed-reordered} does not quite work (recall that we relied on UIP in \cref{subsec:GAT}).
Let us start by formulating representability.

\begin{definition}
 Let $F : \Acat \to \TT$ be an $\infty$-functor.
 A \emph{representation} for $F$ is a tuple $(x, u, r)$ where $x : \Acat_0$, $u : F_0 \, x$, and $r : (y : \DD_0) \to (v : F_0 \, y) \to \iscontr\left(\Sigma(f : \Acat_1 \, x \, y). F_1 \, f \, u = v\right)$.
\end{definition}

Applying this definition to our case of interest leads us to:

\begin{definition}[context extension structure] \label{def:conext}
    Let an $\infty$-category $\CC$ be given together with $\infty$-functors $\Ty : \CCop \to \TT$ and $\Tm : \left(\int_{\CCop} \Ty\right) \to \TT$.
    A \emph{context extension structure} consists, for all $\Gamma : \CC_0$ and $A : \Ty_0 \, \Gamma$, of the following data:
    \begin{enumerate}
        \item an object $\Gamma \ext A : \CC_0$,
        \item a morphism $\p_{A} : \CC_1 \, (\Gamma \ext A) \, \Gamma$,
        \item and a term $\q_{A} : \Tm_0 ((\Gamma \ext A) , (\Ty_1 \, \p_{A} \, A))$.
    \end{enumerate}
    Whenever, in addition to $\Gamma$ and $A$, we have $\Theta : \CC_0$ 
    and $\sigma : \CC_1 \, \Theta \, \Gamma$ 
    and $t : \Tm_0 (\Theta , (\Ty_1 \, \sigma \, A))$, then 
    we also have the following data:
    \begin{enumerate}
        \setcounter{enumi}{3}
        \item a morphism $(\sigma \comma t) : \CC_1 \, \Theta \, (\Gamma \ext A)$,
        \item a triangle filler $\ext\upbeta_1^{A,\sigma,t} : \CC_2 \, (\sigma \comma t) \, \p_{A} \, \sigma$
        \item an equality $\ext\upbeta_2^{A,\sigma,t} : \Tm_1 \, ((\sigma \comma t) , \refl) \, \q_A = t$ over the equality we get from $\Ty_2 \, \ext\upbeta_1^{A,\sigma,t}$.
        \item for any 3-tuple $(\tau, f, e)$ of the same type as $((\sigma \comma t), \ext\upbeta_1^{A,\sigma,t}, \ext\upbeta_2^{A,\sigma,t})$, the two tuples are equal.
    \end{enumerate}
\end{definition}    

\begin{takeaway}
    \noindent
    \textbf{Section summary.}
    An $\infty$-CwF has the following components:
    \begin{enumerate}
        \item An $\infty$-category $\CC$, cf.\ \cref{def:infcategory}. This is a semisimplicial type $(\CC_0, \CC_1, \CC_2, \ldots)$ such that every inner horn has a contractible type of fillers and, for every object, we have an idempotent equivalence.
        \item $\CC$ has a terminal object, cf.\ \cref{def:terminal}. An object $x$ is terminal if every $\CC_1 \, y \, x$ is contractible.
        \item An $\infty$-functor $\Ty : \CCop \to \TT$.
        Here, $\TT$ (cf.\ \cref{def:TT}) is the semisimplicial type of (small) types and functions, with $\TT_0 \equiv \UU$ and $\TT_1 \,X \, Y \defeq (X \to Y)$.
        An $\infty$-functor (cf.\ \cref{def:ifunctor}) is a natural transformation between semisimplicial types, i.e.\ a sequence $(F_0, F_1, F_2, \ldots)$, which maps identities to equivalences.
        \item A second $\infty$-functor $\Ty : \int_{\CCop} \Ty \to \TT$. The $\infty$-category of elements (cf.\ \cref{def:catofels}) is defined 
        by first constructing the $\infty$-category of pointed types.
        \item A context extension structure, cf.\ \cref{def:conext}. Context extension can be represented in a finite way and is essentially the same as in the UIP case.
    \end{enumerate}
    Christian Sattler has suggested an alternative where the above functors $\Ty$ and $\Tm$ are replace by a sequence of right fibrations, $\Tm \twoheadrightarrow \Ty \twoheadrightarrow \CC$.
    Context extension can then be described as a right adjoint of the first fibration, in line with Awodey's presentation~\cite{awodey_natural}.
\end{takeaway}

\section{Examples of $\boldsymbol{\infty}$-Categories with Families} \label{sec:examples-of-iCwF}

\subsection{The Syntax as a QIIT} \label{subsec:syntax-as-hcwf}

We will discuss later (\cref{thm:set-cwf-is-set-cwf}) that every set-CwF in the sense of \cref{def:setCwF} and \cref{fig:Ambrus-changed-reordered} can be presented as an $\infty$-CwF.
Therefore, the \emph{syntax QIIT} by Altenkirch and Kaposi~\cite{alt-kap:tt-in-tt}, discussed in \cref{ex:initial}, is an $\infty$-CwF.

The concrete construction of the $\infty$-CwF $\CC$ from \cref{fig:Ambrus-changed-reordered} can be described as follows. One starts by defining $\CC_0$ to be $\Con$ and by setting $\CC_1 \, \Gamma \, \Delta$ to be $\Sub \, \Gamma \, \Delta$. The next level is given by $\CC_2 \, f \, g \, h \defeq (g \mc f = h)$, and all higher components of $\CC$ are contractible: $\CC_{3 + n} \blank \defeq \unit$. The other parts are constructed analogously.

It is possible to define what a morphism between $\infty$-CwF's is, and thereby to state what it means for an $\infty$-CwF $\CC$ to be \emph{initial} -- to be understood in the usual sense of \emph{homotopy initial}, i.e.\ there is a contractible type of morphisms from $\CC$ to any other given $\infty$-CwF.

While the syntax in \cref{ex:initial} is initial among set-CwF's, it is highly unclear whether, if viewed as an $\infty$-CwF (\cref{subsec:syntax-as-hcwf}), it is also initial among higher CwF's.
Instead, we can attempt to define the initial such $\infty$-CwF directly in the next section.

\subsection{The Initial Model as a HIIT} \label{subsec:initial-HIIT}

The definition of an $\infty$-CwF is fully algebraic in the sense of Cartmell~\cite{cartmell1986generalised}, although with infinitely many sorts, operations, and equations, and this is still true if we add components such as base types or $\Pi$-types.
Clearly, semisimplicial types themselves are an (infinite) generalised algebraic theory, as seen from the presentation (\ref{eq:sst1}--\ref{eq:sst3}).

For the Segal condition, this is not as obvious, but still easy to see; 
in fact, we have in \cref{subsec:highersemicats} 
chosen the specific formulation of the Segal condition via horn filling in order to simplify this part.

Assume we have a semisimplicial type $B \equiv (B_0, B_1, B_2, \ldots)$.
We can generate the 
\emph{localisation}%
\footnote{Localisations in homotopy type theory have been studied by Rijke, Shulman, and Spitters~\cite{rijke2017modalities} as well as by Christensen, Opie, Rijke, and Scoccola~\cite{christensen2018localization}.}
of $B$ at inner horn inclusions, i.e.\ the 
\emph{free $\infty$-semicategory}  
$A \equiv (A_0, A_1, A_2, \ldots)$ generated by $B$ via a huge higher inductive-inductive construction.
First, we want to ensure that $A$ contains $B$, so we unsurprisingly start with:
\begin{alignat}{3}
& \eta_0 && : &\quad & B_0 \to A_0 \\
& \eta_1 && : && (x_0 \, x_1 : B_0) \to (B_1 \, x_0 \, x_1) \to A_1 (\eta_0 \, x_0) (\eta_0 \, x_1) \\
& \eta_2 && : && \{x_0 \, x_1 \, x_2 : B_0\} \to
(x_{01} : B_1 \, x_0 \, x_1) \to 
(x_{12} : B_1 \, x_1 \, x_2) \to \notag \\
&&&&& (x_{02} : B_1 \, x_0 \, x_2) \to 
(B_2 \, x_{01} \, x_{12} \, x_{02}) \to 
A_2 \, (\eta_1 \, x_{01}) (\eta_1 \, x_{12})  (\eta_1 \, x_{02})
\end{alignat}

\noindent
Next, each level of the Segal-condition needs to be split into four parts.
For example, to state that $\Lambda^3_1$-horns have contractible types of fillers, we need (we omit implicit arguments):

\begin{alignat}{3}
 & h_2^e && : &\quad & 
 (x_{01} : A_1 \, x_0 \, x_1) \to
 (x_{12} : A_1 \, x_1 \, x_2) \to
 A_1 \, x_0 \, x_2 \\
 & h_3^e && : && 
 (x_{01} : A_1 \, x_0 \, x_1) \to
 (x_{12} : A_1 \, x_1 \, x_2) \to
 A_2 \, x_{01} \, x_{12} \, (h_2^e \, x_{01} \, x_{12})\\
 & h_2^u && : &&
 (x_{01} : A_1 \, x_0 \, x_1) \to
 (x_{12} : A_1 \, x_1 \, x_2) \to \notag \\
 &&&&& 
 (f : A_1 \, x_0 \, x_2) \to
 (A_2 \, x_{01} \, x_{12} \, f) \to
 f = (h_2^e \, x_{01} \, x_{12}) \\
 & h_3^u && : && 
 (x_{01} : A_1 \, x_0 \, x_1) \to
(x_{12} : A_1 \, x_1 \, x_2) \to \notag \\
&&&&& 
(f : A_1 \, x_0 \, x_2) \to
(F : A_2 \, x_{01} \, x_{12} \, f) \to
F =_{h_2^u} (h_3^e \, x_{01} \, x_{12})
\end{alignat}

\noindent
The intuition behind the name $h_i^{e/u}$ above is as follows: $i$ is the level, $e$ stands for \emph{existence} [of a filler], $u$ stands for \emph{uniqueness} [of the existing filler].

The other components of an $\infty$-CwF can be written as a GAT as well.
If we assume that our host type theory has an infinite version of higher inductive-inductive types as specified by Kaposi and Kov\'acs~\cite{kaposi2019signatures}, then this gives us the initial $\infty$-CwF.
The semisimplicial type $B$ can be used to specify base types, or it can be taken to be empty with base types added as part of the induction.

We do not know under which conditions this HIIT turns out to have decidable equality,
but if it does,
then it in particular is set-based in the sense of \cref{subsec:setbased} below.
This means that it would be equivalent to a set-CwF (cf.\ \cref{thm:set-cwf-is-set-cwf}), namely the syntax QIIT.

\subsection{Higher Models from Strict Models and the Standard Interpretation}

We show that, from a \emph{strict} CwF such as the standard model, we get an $\infty$-CwF.

\begin{definition}[strict CwF]
    In 2LTT, a \emph{strict category with families (sCwF)} is a CwF as in \cref{fig:Ambrus-changed-reordered} such that $\Con$, $\Sub$, $\Ty$ and $\Tm$ are all fibrant types, while all the stated equalities for contexts, substitutions, types, and terms hold up to strict equality ($\steq$).
    For context extension, we need the contractibility condition \eqref{eq:contr-represent}, stated using the usual fibrant equality type ($=$).
\end{definition}

\begin{lemma} \label{lem:strict-gives-inf}
    Given a sCwF, we can construct an $\infty$-CwF.
\end{lemma}
\begin{proof}[Proof sketch] \let\qed\relax
The first part of the proof is given by the construction of semi-Segal types from strict categories, as described in \cite{ann-cap-kra:two-level} and sketched in \cref{subsec:types}.
To construct the remaining components, we first generate strict semisimplicial diagrams by taking nerves $\nerve$ and $\nervep$ as in \cref{subsec:terms}.
\end{proof} 

\vspace*{-.15cm}

\noindent
\begin{minipage}[t]{\textwidth -4.2cm}

\hspace*{.3cm }The Reedy fibrant replacement operation $R$ constructed in~\cite{ann-cap-kra:two-level} has the property that, from a strict natural transformation $\tau : X \to Y$ between strict diagrams over $\deltop$, we get a strict natural transformation $R(\tau) : R(X) \to R(Y)$ between the respective Reedy fibrant replacements, i.e.\  semisimplicial types. 
This follows from \cite[Cor 4.28]{ann-cap-kra:two-level}, applied on the diagram on the right. 
\end{minipage}
\hspace{\fill}
\begin{minipage}[t]{4cm}
    \vspace*{-.5cm}
    
\begin{equation*}
\begin{tikzpicture}[x= 2.5cm , y = -1cm]
\node (X) at (0,0) {$X$};
\node (RX) at (0,1) {$R(X)$};
\node (OneX) at (0,2) {$1$};
\node (Y) at (1,0) {$Y$};
\node (RY) at (1,1) {$R(Y)$};
\node (OneY) at (1,2) {$1$};

\draw[->] (X) to node[left] {} (RX);
\draw[->>] (RX) to node[left] {} (OneX);
\draw[->] (Y) to node[right] {} (RY);
\draw[->>] (RY) to node[right] {} (OneY);

\draw[->] (X) to node[above] {$\tau$} (Y);
\draw[->, dashed] (RX) to node[above] {$R(\tau)$} (RY);
\draw[->] (OneX) to node[above] {} (OneY);

\end{tikzpicture}
\end{equation*}    
\end{minipage}

\vspace*{.1cm}

Simply fibrantly replacing everything does not fully work for the $\infty$-functor $\Tm$ since the fibrant replacement does not commute with the required strict pullback, 
but switching to right fibrations and back (as discussed in \cref{rem:left-fib-classifier}) resolves the issue.
Context extension can be checked manually, since it only concerns the lowest levels. \qed

\vspace*{.2cm}
Applying \cref{lem:strict-gives-inf} on the 1-categorical standard model (which is a strict CwF, see \cref{ex:stdmodel}) shows that the standard model is also an $\infty$-CwF.

\subsection{Slicing in $\boldsymbol{\infty}$-CwF's}

As a final example, let us show that the slice construction analogous to \cref{ex:axmodel} works for $\infty$-CwF's. Slicing of $\infty$-categories in type theory works in essentially the same way as slicing in Kan simplical sets or similar settings \cite{lurie:higher-topoi}.

Given an $\infty$-category $\CC$ with an object $\Gamma : \CC_0$, we want to construct the \emph{slice $\infty$-category} $(\CC \slash \Gamma)$.
Recall that elements of $\CC_n$ can, by Yoneda, be seen as $n$-simplexes $\simplex n \functor \CC$.
We define $(\CC \slash \Gamma)_n$ to be the type of $(n+1)$-simplexes where the last vertex strictly equals $\Gamma$; in other words, $(\CC \slash \Gamma)_{n}$ is defined to be $\CC_{n+1}$ with a condition involving a strict equality.
This type is fibrant nevertheless:
$(\CC \slash \Gamma)_n$ is (strictly isomorphic to) a nested $\Sigma$-type of $(n+2)$ vertexes including the last vertex $x_{n+1} : \CC_0$, lines, triangle fillers, \ldots, and the strict equality $e : x_{n+1} = \Gamma$. But the type of pairs $(x_{n+1}, e)$ is strictly isomorphic to the unit type and therefore fibrant.

The morphism part of $(\CC \slash \Gamma)$ is given as those morphisms in $\CC_{n+1}$ which do not touch the last vertex. Thinking of $\CC_{n+1}$ as simplexes and face maps as projections as in \eqref{eq:sst-as-functor}, face maps are those projections which do not remove the last vertex $x_{n+1}$.

One can check that $(\CC \slash \Gamma)$ is a semi-Segal type. Further, if $i$ is an identity on $\Theta$ in $\CC$ and $f : \CC_1 \, \Theta \, \Gamma$ an object of $(\CC \slash \Gamma)$, then the identity on this object is given by the cell in $\CC_2 \, f \, i \, f$ of \cref{lem:id-characterisation}. Thus, $(\CC \slash \Gamma)$ is an $\infty$-category.
As in the 1-categorical case, the functors $\Ty$ and $\Tm$ as well as context extension are not affected by the slicing operation.

\begin{takeaway}
    \noindent
    \textbf{Section summary.}
    In \cref{sec:challenges}, we have seen several examples that show why the notion of a CwF implemented in type theory without UIP is unsatisfactory.
    In the current section, we have seen that all these examples work as expected as soon as we move to $\infty$-CwF's.
\end{takeaway}

\section{Variations: Set-Based, Univalent, and Finite-Dimensional Models} \label{sec:variations}

We introduce three possible properties that an $\infty$-CwF can have.
Each of these three properties is a proposition.

\subsection{Set-Based $\boldsymbol{\infty}$-CwF's} \label{subsec:setbased}

Recall that a type is said to be a \emph{set} if it satisfies UIP.
\begin{definition}[set-based $\infty$-CwF]
    An $\infty$-CwF is \emph{set-based} if, for all $\Gamma, \Delta, A$, each of the types $\CC_0$, $\CC_1 \, \Gamma \, \Delta$, $\Ty_0 \, \Gamma$, and $\Tm_0 (\Gamma,A)$ is a set.
\end{definition}

This can in more detail be described as:
\begin{lemma} \label{lem:set-based-implies}
    In a set-based $\infty$-Cwf, \emph{all} of the following apply:
    \begin{itemize}
        \item By definition, the types and type families $\CC_0$, $\CC_1$, $\Ty_0$, and $\Tm_0$ are all set-truncated;
        \item it follows that $\CC_2$, $\Ty_1$, $\Tm_1$ are families of propositions;
        \item and all remaining occurring components ($\CC_{3+i}$, $\Ty_{2+i}$, $\Tm_{2+i}$) are families of contractible types.
    \end{itemize}
\end{lemma}
\begin{proof}
    For $\CC$, this follows from the Segal condition in the same way as in \cite[Lem 8.9.3]{nicolai:thesis}.
    For $\Ty$ and $\Tm$, this is follows from the fact that $\TT$, $\TTp$ have the corresponding property.
\end{proof}

We can now conclude:

\begin{theorem} \label{thm:set-cwf-is-set-cwf}
    Set-based $\infty$-CwF's are equivalent to the set-CwF's of \cref{def:setCwF}.
\end{theorem}
\begin{proof}
    The above lemma shows that almost all components of a set-based $\infty$-CwF are trivial. Together with 
    the translation  by Capriotti and the author~\cite{capKra_semisegal}, this proves that the data of the semi-Segal  type is equivalent to the representation of a semicategory in \cref{fig:Ambrus-changed-reordered}.
    All remaining parts are easy to check.
\end{proof}

Note that \cref{thm:set-cwf-is-set-cwf} also implies that the formalised syntax using UIP (in particular the \emph{syntax QIIT})
can be formulated as an $\infty$-CwF (cf.\ \cref{subsec:syntax-as-hcwf}).

\begin{example} \label{ex:std-is-set-based}
    The \emph{syntax QIIT} by Altenkirch and Kaposi~\cite{alt-kap:tt-in-tt} is an example for a set-based $\infty$-CwF, and a non-example is the standard model as in \cref{ex:stdmodel}.
\end{example}

\subsection{Univalent $\boldsymbol{\infty}$-CwF's}

Since an identity in an $\infty$-CwF $\CC$ is by definition an idempotent equivalence, we have, for all $\Gamma, \Delta : \CC_0$, an obvious map
\begin{equation} \label{eq:id2eqv}
  (\Gamma = \Delta) \to \Sigma(e : \CC_1 \, \Gamma \, \Delta). \iseqv(e).
\end{equation}
The consequent formulation of univalence is standard:
\begin{definition}[{univalent $\infty$-CwF}]
An $\infty$-category is \emph{univalent} if the function \eqref{eq:id2eqv} is an equivalence.
\end{definition}

This allows us to connect our version of higher categories to the complete semi-Segal types by Capriotti~\cite{paolo:thesis} and others.
Recall that a semi-Segal type is \emph{complete} if, for every object $\Gamma$, the type $\Sigma(\Delta : \CC_0). \Sigma(e : \CC_1 \, \Gamma \, \Delta). \iseqv(e)$ is contractible.
\begin{lemma}
    A semi-Segal type $A$ is a univalent $\infty$-category if and only if it is \emph{complete} in the sense of Capriotti~\cite{paolo:thesis}.
\end{lemma}
\begin{proof}
    From a complete semi-Segal type, we can construct an explicit identity structure as in \eqref{eq:I-def} using \cref{lem:I-is-id}.
    Given the identity structures, univalence and completeness both say that the total space $\Sigma(e : \CC_1 \, \Gamma \, \Delta). \iseqv(e)$ is contractible, making them equivalent.
\end{proof}

\begin{theorem} \label{thm:set-plus-univ-weird}
    In an $\infty$-CwF that is both set-based and univalent,
    the identities are the only auto-equivalences.
\end{theorem}
\begin{proof}
    By univalence, equivalences and equalities coincide; and by assumption, $\refl$ is the only equality of type $\Gamma = \Gamma$.
\end{proof}

\cref{thm:set-plus-univ-weird} implies that only trivial models can be set-based and univalent at the same time, and interesting set-based representations of the syntax are definitely not univalent.
Already the simple context $(b : A, c : A)$ has the non-trivial auto-equivalence which swaps $b$ and $c$.

\begin{example}
    The standard model in \cref{ex:stdmodel} is a univalent $\infty$-CwF, while the \emph{syntax QIIT} by Altenkirch and Kaposi~\cite{alt-kap:tt-in-tt} is not univalent.
\end{example}

\subsection{Finite-Dimensional Models}

We have seen that set-based CwF's correspond to $1$-categories.
From the point of view of homotopy type theory,
it is somewhat unusual that objects and morphisms ($\CC_0 $ and $\CC_1$) are restricted to the same truncation level.
Thus, we say that an $\infty$-CwF are \emph{of dimension $1$} if $\CC_0$ is a 1-type and $\CC_1$ is a family of sets.
We also simply say that $\CC$ is a $1$-CwF.
The canonical generalisation is the following:

\begin{definition}[{finite-dimensional $\infty$-CwF} \emph{or} $n$-CwF]
    Let $n : \sN$ be a number.
    An $\infty$-CwF is of dimension $n$ if $\CC_0$ is an $n$-type, and $\CC_1$, $\Ty_0$, as well as $\Tm_0$ are families of $(n-1)$-types.
    In this case, we say that $\CC$ is an \emph{$n$-CwF}.
\end{definition}

Note that the condition on $\Ty_0$ amounts to requiring that $\Ty : \CCop \rffunctor \UU$ factors through $\UU^{\leq n} \defeq \Sigma(X : \UU). \mathsf{is}\mbox{-}n\mbox{-}\mathsf{type}(X)$, the semi-Segal type of $n$-types.
The analogous observation holds for $\Tm_0$.

\begin{lemma} \label{lem:n-dim-implies}
    If an $\infty$-Cwf is of dimension $n$, then we have:
    \begin{itemize}
        \item For all $i : \sN$, the type family $\CC_{n+i}$ is $(n-i)$-truncated (or contractible, if $n-i \leq -2$);
        \item $\Ty_0$ is an $(n-1)$-truncated family and, for all $i \geq 1$, $\Ty_{n+i}$ is $(n-i)$-truncated;
        \item and finally, $\Tm_0$ is an $(n-1)$-truncated family, while $\Tm_{n+i}$ are $(n-i)$-truncated.
    \end{itemize}
\end{lemma}
\begin{proof}
    Analogous to \cref{lem:set-based-implies}.
\end{proof}

\begin{remark}
    The awkward special cases for $\Ty_0$ and $\Tm_0$ stem from the fact that we are working with functors into $\TT$. If we work with right fibrations (see \cref{rem:left-fib-classifier}) instead, the truncation conditions are more regular.
\end{remark}

As in the set-based case, almost all components of an $n$-dimensional CwF are contractible families and can be omitted since they do not carry any information. The only exception to this is $\CC_{n+2}$, which is a contractible family but which is required to formulate the Segal condition on the level that affects $\CC_{n+1}$, i.e.\ non-contractible components depends on it.%
\footnote{This is identical to the situation encountered by Capriotti and the author~\cite{capKra_semisegal} who show that a univalent 1-category corresponds to a semi-Segal type $(A_0, A_1, A_2, A_3)$ with structure,
where it is necessary to include $A_3$ even though it is always contractible (unless, of course, one wants to replace it by another special case).}
Each of the finitely many components can be formulated purely in the inner/fibrant fragment of 2LTT.
This observation implies:

\begin{theorem} \label{thm:finite-fibrant}
    In ``plain'' homotopy type theory \cite{hott-book}, and for an externally fixed natural number $n$, we can define a type of $n$-CwF's.
    In 2LTT (even without the axiom that $\sN$ is cofibrant), for a number $n : \sN$, the type of $n$-CwF's is fibrant.
\end{theorem}
\begin{proof}
    This follows from a result by Shulman~\cite[Lem 11.8]{shulman_inversediagrams} and the analogous formulation in 2LTT~\cite[Thm 4.8]{ann-cap-kra:two-level}.
\end{proof}

\begin{example}
    Since $n$-CwF's are a generalisation of set-based $\infty$-CwF, what we said in \cref{ex:std-is-set-based} can be transferred: The \emph{syntax QIIT} \cite{alt-kap:tt-in-tt} is an $n$-CwF.
    The standard model from \cref{ex:stdmodel} is not an $n$-CwF for any $n : \sN$, but by replacing the universe $\UU$ in its construction by the type $\UU^{\leq n}$, we get the $(n+1)$-CwF of $n$-types.
\end{example}

\begin{takeaway}
    \noindent
    \textbf{Section summary.}
    We have identified three possible additional properties that an $\infty$-CwF can have. It can be \emph{univalent}, which is natural in homotopy type theory; an obvious example is the standard model.
    An $\infty$-CwF can further be \emph{finite-dimensional}.
    This gives rise to the notion of an \emph{$n$-CwF} which can be expressed in homotopy type theory and vastly generalises the internal CwF's that the current literature considers.
    In particular, this notion is general enough to allow non-trivial versions of the standard model, albeit restricted to $n$-types.
    \emph{Set-based $\infty$-CwF's} are a very weak special case and even less general than $1$-CwF's, but a notion that captures many developments that have been done assuming UIP.
\end{takeaway}

\section{Open Problems and Future Directions}
\label{sec:openconclusions}

We have demonstrated that higher-dimensional categories lead to a notion of model of type theory that is very well-behaved and works in situations where 1-categorical models are unsuitable.
Our work inspires numerous new questions.
First of all, it seems natural to equip the definition of an $\infty$-CwF with additional components such as $\Pi$-types, $\Sigma$-types, universes, or simple base types.
Developing a notion of \emph{$\infty$-natural transformation}, it is easy to define what a \emph{morphism} between $\infty$-categories is; but the natural expectation would be that [small] $\infty$-CwF's form a [large] $(\infty,2)$-category, and it is much less clear how this can be formulated. The intermediate goal would be to show that $\infty$-CwF's form a large $\infty$-CwF.

A very important question is whether the syntax (defined as a quotient inductive-inductive type) is initial as an $\infty$-CwF (after adding base types and other type formers). Phrased differently, we can ask whether the initial $\infty$-CwF has decidable equality.
This is very closely related to the open problem whether HoTT can ``eat'' itself (Shulman~\cite{shulman:eating}).
Our conjecture is that HoTT cannot eat itself, but that 2LTT can eat itself.
Of course, the intermediate goal that this paper is working towards is the statement that 2LTT can eat HoTT.
Even an even much weaker question, namely whether the initial $\infty$-CwF has trivial fundamental group, would be very interesting; but even this seemingly much simpler problem appears to be highly non-trivial, and similar results for seemingly simpler situations~\cite{kraus_FHG,krausVonRaumer:wellfounded} suggest that a solution will require new techniques.

If it turns out that the (set-truncated) syntax is indeed the initial $\infty$-CwF it would show that, if we give ourselves a way to talk about semisimplicial types, then HoTT can eat itself.
The other direction seems much more in reach: If HoTT can eat itself, then semisimplicial types can be defined. This has already been conjectured by Shulman~\cite{shulman:eating} but not been made precise.

Another question is whether the higher inductive-inductive definition described in \cref{subsec:initial-HIIT} really requires what is described as \emph{HIITs with infinitely many components}.
We would expect that such HIITs are not required and ordinary HIITs, albeit indexed over $\sN$, suffice to encode the desired initial model.
It also would be interesting to formulate and study \emph{higher categories with attributes} (which one would expect to be equivalent to $\infty$-CwF's), the more general \emph{higher comprehension categories}, and other internal $\infty$-categorical formulations of models that have been considered in 1-category theory.

The connection to \emph{directed type theory}~\cite{LICATA2011263,rhiel-shulman:directed,nuyts:masterthesis,NORTH2019223} is intriguing.
Future directed type theories may allows us to develop $\infty$-category theory in a synthetic way, and the relationship to the half-synthetic approach of this paper will be interesting to see.

\subsection*{Acknowledgements}

This paper has benefited from many conversations that I had with various people during the last couple of years.

I am particularly grateful for the numerous discussions with Christian Sattler and Paolo Capriotti.
I have learned many $\infty$-categorical ideas from Christian and Paolo, and their support allowed me to understand much more of the higher categorical literature than I would have managed to understand on my own.
Moreover, Paolo pointed me to abstract and alternative definitions of CwF's and representability, which proved to be very useful for the development in this paper.

I am also thankful for very helpful comments on an earlier draft version of this paper that I have received from Christian and anonymously.

Additional special thanks go to Ambrus Kaposi and Thorsten Altenkirch for explaining me their intrinsically well-typed syntax (the \emph{syntax QIIT}) until the main points sunk in.

I am grateful to the research communities in Birmingham, Budapest, Nottingham, and Pittsburgh (CMU) for the opportunities to give talks and discuss the main ideas of this paper.
These opportunities have led to valuable feedback and various improvements.

\bibliographystyle{ieeetr}
\bibliography{master}

\end{document}